\documentclass[lettersize,journal]{IEEEtran}
\IEEEoverridecommandlockouts

\usepackage{amsmath,amsfonts}
\usepackage{amsthm,amssymb}
\usepackage{braket}
\usepackage{algorithmic}
\usepackage{array}

\usepackage{booktabs}
\usepackage{textcomp}
\usepackage{caption}
\usepackage{subcaption}
\usepackage{stfloats}
\usepackage{url}
\usepackage{xcolor}
\usepackage{tcolorbox}
\usepackage{verbatim}
\usepackage{graphicx}
\usepackage{makecell}
\usepackage{bm}
\usepackage{hyperref}

\newtheorem{definition}{Definition}

\newtheorem{theorem}{Theorem}
\newtheorem{lemma}{Lemma}

\usepackage{qcircuit}
\usepackage{environ}
\newcommand{\myqctmp}[2][0.25]{\Qcircuit @C=#2em @R=#1em @!R}
\NewEnviron{myqcircuit}[1][0.25]{\vcenter{\myqctmp[#1]{0.75} {\BODY}}}
\NewEnviron{myqcircuit*}[2]{\vcenter{\myqctmp[#1]{#2} {\BODY}}}
\newcommand{\controlsq}{*!<0em,.025em>-=-<0em>{\square}} 
\newcommand{\ctrlsq}[1]{\controlsq \qwx[#1] \qw}         

\def\BibTeX{{\rm B\kern-.05em{\sc i\kern-.025em b}\kern-.08em
    T\kern-.1667em\lower.7ex\hbox{E}\kern-.125emX}}
\usepackage{balance}

\usepackage[normalem]{ulem}

\definecolor{revhl}{RGB}{255,242,153}    

\usepackage{colortbl}

\usepackage{soul}
\sethlcolor{revhl}                       
\soulregister{\textbf}{1}
\soulregister{\emph}{1}
\soulregister{\textit}{1}
\soulregister{\cite}{1}
\soulregister{\ref}{1}

\usepackage[framemethod=TikZ]{mdframed}
\newmdenv[
  backgroundcolor=revhl,                 
  linecolor=revhl,
  innertopmargin=2pt,
  innerbottommargin=4pt,
  innerleftmargin=6pt,
  innerrightmargin=6pt,
  skipabove=4pt,
  skipbelow=4pt
]{hiblock}

\begin{document}
\title{Reducing C-NOT Counts for State Preparation and Block Encoding via Diagonal Matrix Migration
}

\author{
\thanks{Corresponding authors: Guofeng Zhang (guofeng.zhang@polyu.edu.hk) and Xiao-Ming Zhang (phyxmz@gmail.com)}
Zexian Li, Guofeng Zhang \IEEEmembership{Senior Member, IEEE}, and Xiao-Ming Zhang\thanks{Zexian Li and Guofeng Zhang are with the Department of Applied Mathematics, The Hong Kong Polytechnic University, Hong Kong, China (email:zexian.li@connect.polyu.hk; guofeng.zhang@polyu.edu.hk), and Guofeng Zhang is also with The Hong Kong Polytechnic University Shenzhen Research Institute, Shenzhen, Guang Dong 518057, China}\thanks{Xiao-Ming Zhang is with Key Laboratory of Atomic and Subatomic Structure and Quantum Control (Ministry of Education), South China Normal University, Guangzhou, China, and Guangdong Provincial Key Laboratory of Quantum Engineering and Quantum Materials, Guangdong-Hong Kong Joint Laboratory of Quantum Matter, South
China Normal University, Guangzhou, China. (e-mail:phyxmz@gmail.com)}}

\markboth{Journal of \LaTeX\ Class Files,~Vol.~18, No.~9, September~2020}%
{How to Use the IEEEtran \LaTeX \ Templates}

\maketitle

\begin{abstract}
    Quantum state preparation and block encoding are versatile and practical input models for quantum algorithms in scientific computing. The circuit complexity of state preparation and block encoding frequently dominates the end-to-end gate complexity of quantum algorithms. We give algorithms with lower C-NOT counts for both the state preparation and block encoding. For a general $n$-qubit state, we improve the C-NOT count of the Plesch–Brukner algorithm (2011) from  $(23/24)\times 2^n$ to $(11/12)\times2^n$. For block encoding, our single-ancilla protocol for $2^{n-1}\times 2^{n-1}$ matrices uses the spectral norm as subnormalization and achieves a C-NOT count leading term $(11/48)\times 4^n$. Further optimization is performed for low-rank matrices, which frequently arise in practical applications. Specifically, we achieve the C-NOT count leading term $(2^{\lceil\log_{2}K\rceil}+(11/12))\times 2^n$ for a rank-$K$ matrix. This is the first quantum algorithm that encodes matrices using the optimal normalization factor while also allowing the C-NOT count to be adjusted according to the matrix rank. Our approach builds upon the recursive Block-ZXZ decomposition from Krol et al. (2024) and introduces a diagonal matrix migration technique based on the commutativity of the diagonal matrix and the uniformly controlled rotation about the $z$-axis to minimize the use of C-NOT gates.

\end{abstract}


\section{Introduction}
\label{section_Introduction}

\IEEEPARstart{Q}{uantum} algorithms offer significant potential for addressing complex problems in scientific computing. The fundamental mathematical structures underlying these scientific inquiries encompass vectors, matrices, and tensors. Vectors and matrices are the most common forms; a vector can usually be accessed by the quantum state preparation~\cite{10.1109/TCAD.2023.3297972,10.1109/TCAD.2023.3244885,PhysRevLett.129.230504}, while a matrix can be mainly accessed via qubitization~\cite{Low2019hamiltonian}, block encoding~\cite{10.1145/3313276.3316366}, the sparse-access input model~\cite{Berry2007,PhysRevX.13.041041}, linear combinations of unitary operations~\cite{10.5555/2481569.2481570,Chakraborty2024implementingany} and Trotter splitting~\cite{doi:10.1126/science.273.5278.1073,PRXQuantum.6.010359,PhysRevResearch.6.033147}, and others~\cite{PhysRevLett.131.150603,cvl9-97qg}. Among the aforementioned methods for matrix access, block encoding is recognized as one of the most flexible input models for linear systems~\cite{PRXQuantum.3.040303,low2024quantumlinearalgorithmoptimal,somma2025quantumalgorithmlinearmatrix}.

End-to-end complexity analysis~\cite{Dalzell_McArdle_Berta_Bienias_Chen_Gilyén_Hann_Kastoryano_Khabiboulline_Kubica_etal._2025} is necessary to determine if a quantum algorithm theoretically outperforms classical ones. One of the convincing metrics in the end-to-end complexity analysis is to quantify the circuit size or the gate counts required by the algorithm. In this work, the gate counts for state preparation and block encoding are considered.

\subsection{Related work}

There exists an inherent trade-off between the amount of classical precomputation and the resulting C-NOT count in state-preparation circuits. For example, Zenchuk et al. proposed a probabilistic state-preparation method~\cite{zenchuk2025arbitrarystatecreationcontrolled} based on controlled measurements that requires no classical precomputation, albeit at the cost of more than $\mathcal{O}(n 2^n)$ multi-controlled gates.  Yet, measurement-based methods may require multiple repetitions. Moreover, its counterpart, unitary preparation, is more friendly for some quantum subroutines, like linear combination of unitaries (LCU)~\cite{LCU2012} and quantum singular value transformation. In 2004, M\"ott\"onen et al. proposed UCG-based constructions~\cite{PhysRevLett.93.130502} for arbitrary $n$-qubit dense state vectors that require nearly $2^{n+1}$ C-NOTs together with $\mathcal{O}(n 2^n)$ classical precomputation to decouple the uniformly controlled gates~\cite{PhysRevLett.93.130502}. The same decoupling can alternatively be realized using multicontrolled SU(2) gates~\cite{10.1109/TCAD.2023.3327102,8blx-nfcr}, which need only $\mathcal{O}(n 2^n)$ C-NOTs at the cost of $\mathcal{O}(2^n)$ classical precomputation. However, in many quantum algorithms such as the quantum singular value transformation~\cite{10.1145/3313276.3316366}, Hadamard test~\cite{cqjw-kl8s}, the prepared circuit is executed repeatedly after compilation. Consequently, circuits with lower C-NOT counts are strongly preferred.

The lower bound on the number of C-NOT gates required for quantum state preparation is fundamentally determined by the degrees of freedom inherent in the circuit topology. The leading term for the lower bound on C-NOT gates to realize an arbitrary $n$-qubit state preparation is $\frac{1}{2}2^n$~\cite{PhysRevA.83.032302}, and the minimal C-NOT count to generate a three-qubit state is studied by \ifmmode \check{Z}\else \v{Z}\fi{}nidari\ifmmode \check{c}\else \v{c}\fi{} et al.~\cite{PhysRevA.77.032320} in 2008. The existing methods have not achieved the theoretical lower bound. Table~\ref{tab:C-NOT_count state preparation} summarizes some recently proposed state preparation methods. The
state preparation algorithm proposed by Plesch and Brukner~\cite{PhysRevA.83.032302} (PB) achieves a leading constant of $\frac{23}{24}$ in 2011. The second leading constant has been improved by introducing isometry synthesis by Iten et al.~\cite{PhysRevA.93.032318} in 2016. Furthermore, the low-rank state preparation (LRSP) method by Israel et al.~\cite{10.1109/TCAD.2023.3297972} in 2024 enhances the method through entanglement functions. For our ablation experiments, we also compare the C-NOT counts obtained using the PB and LRSP methods, which rely on the Block-ZXZ decomposition~\cite{PhysRevApplied.22.034019} for unitary synthesis. The trade-off relation is shown in Fig.~\ref{fig: trade-off state preparation}.

\begin{figure}[htbp]
    \centering
    \includegraphics[width=0.95\linewidth]{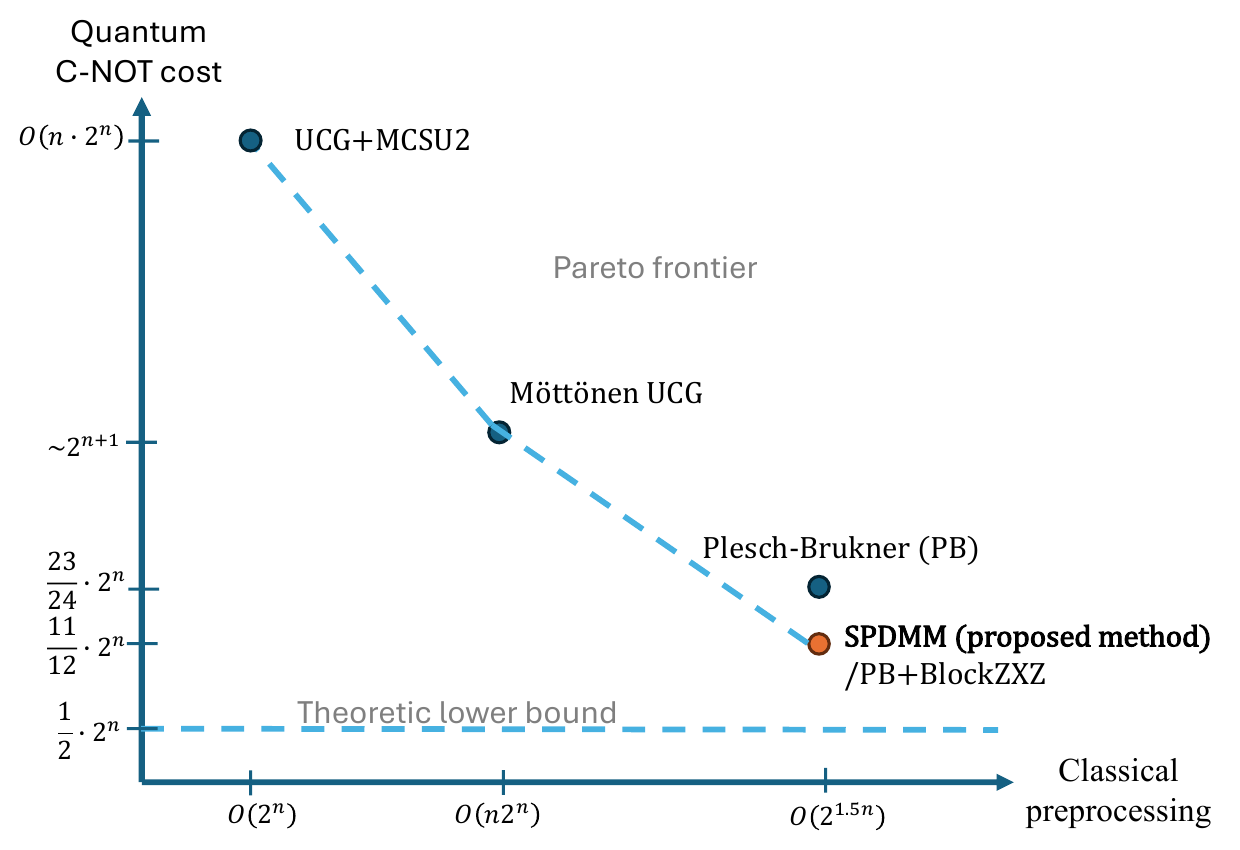}
    \caption{Trade-off between classical preprocessing cost and quantum C-NOT cost for arbitrary $n$-qubit state preparation. The theoretical lower bound $\frac12\cdot 2^n$ is unreachable for ancilla-free methods. The Pareto frontier above comprises the UCG-based constructions\mbox{~\cite{10.5555/2011670.2011675}} with multicontrolled Special Unitary Single-Qubit Gates decomposition\mbox{~\cite{10.1109/TCAD.2023.3327102,8blx-nfcr}} (UCG+MCSU2), the UCG-based constructions with uniformly controlled gate decomposition\mbox{~\cite{10.5555/2011670.2011675}} (M\"ott\"onen UCG), and the state preparation via diagonal matrix migration (SPDMM) or PB methods\mbox{~\cite{PhysRevA.83.032302}} with Block-ZXZ decomposition\mbox{~\cite{PhysRevApplied.22.034019}}. Our proposed SPDMM method lies on the Pareto frontier.}
    \label{fig: trade-off state preparation}
\end{figure}

Block encoding~\cite{10.1145/3313276.3316366} is an input model in which the embedding of a matrix $A$ is typically realized as the leading principal block of a larger unitary acting on a Hilbert space, expressed as
$$
    U = \begin{bmatrix}
        A/\alpha & * \\
        * & *
    \end{bmatrix},
$$
where $*$ denotes irrelevant matrix elements. The matrix must be scaled with a scaling factor such that its spectral norm $\Vert A/\alpha\Vert_2 \leq 1$. Given an arbitrary matrix $A$, existing block encoding protocols have different scaling factors.
Table~\ref{tab:C-NOT_count_SIABLE} summarizes some block encoding protocols.
FABLE~\cite{9951292} is a method to generate approximate circuits for block encoding; when encoding a $2^{n-1}\times 2^{n-1}$ matrix, its scaling factor is proportional to $2^{n-1}$ and the leading term of C-NOT counts is $\frac{1}{2}\times 4^n$. In block encoding protocols using the Frobenius norm as the scaling factor, the C-NOT counts are mainly determined by the state preparation methods~\cite{10.1145/3313276.3316366,10012045,kerenidis_et_al:LIPIcs.ITCS.2017.49,Chakrabarti2020quantumalgorithms}. Among them, the recently proposed BITBLE~\cite{11216010} is one of the compilation-efficient methods.
These protocols have not yet achieved the optimal scaling factor $\alpha=\Vert A\Vert_2$, and this factor is typically proportional to the quantum algorithm's query complexity~\cite{10.1145/3313276.3316366,PRXQuantum.2.040203}. Reducing the scaling factor and C-NOT counts is crucial for optimizing gate complexity in quantum algorithms.

To block encode a matrix with the optimal scaling factor, one of the methods is unitary synthesis. The quantum Shannon decomposition (QSD), proposed by Shende et al.~\cite{1629135} in 2006, achieves a leading term of $\frac{23}{48}4^n$ for an $n$-qubit unitary. In 2024, Krol et al. achieved a leading term of $\frac{11}{24}4^n$ for $n$-qubit decomposition by introducing a Block-ZXZ decomposition~\cite{PhysRevApplied.22.034019}. 
A general block encoding is subject to weaker constraints than unitary synthesis, thereby enabling more efficient implementations. This observation motivates our work to construct a block encoding circuit with a reduced C-NOT count.

Throughout this paper, we take the C-NOT gate count as the primary cost metric for circuit synthesis, following established practice in the field. C-NOT gates are widely regarded as the dominant source of error and the principal experimental bottleneck in near-term quantum hardware, and prior work on circuit and state-preparation synthesis has accordingly focused on minimizing their count~\cite{PhysRevA.83.032302,PhysRevA.93.032318}. This convention is supported quantitatively by current calibration data on every quantum hardware platform in production. The neutral-atom array achieves a two-qubit entangling gate fidelity of $99.5\%$ (infidelity of $5\times 10^{-3}$) on up to 60 atoms in parallel, while the single-qubit fidelity on the same platform is $99.97\%$ (infidelity of $\lesssim 3\times 10^{-4})$ in 2023~\cite{Evered2023}.
By the end of 2025, the state-of-the-art two-qubit gate error is $8.4(7)\times 10^{-5}$~\cite{hughes2025trappediontwoqubitgates9999}, while the single-qubit gate error is $\sim 1.5\times 10^{-7}$~\cite{42w2-6ccy}. In every case, two-qubit gates dominate the physical error budget, and reducing the C-NOT count is the most direct route to reducing total circuit error.

\begin{table*}[t]
    \centering
    \caption{Comparison of the number of C-NOT used between the proposed state preparation method via diagonal matrix migration (SPDMM) and some other state preparation algorithms. The numbers in the first row refer to $n$-qubit states. } \label{tab:C-NOT_count state preparation}
    \begin{tabular}{|ll|cccccccccccccc|}
		\hline
		Methods & Script & 2 & 3 & 4 & 5 & 6 & 7 & 8 & 9 & 10 & 11 & 12 & 13 & 14 & 15 \\
		\hline
        PB~\cite{PhysRevA.83.032302,qclib} & \href{https://github.com/qclib/qclib}{qclib} & 1 & 4 & 9 & 26 & 46 & 126 & 213 & 557 & 919 & 2343 & 3789 & 9581 & 15454 & 38814 \\
        Isometry~\cite{PhysRevA.93.032318,qiskit2024} & \href{https://quantum.cloud.ibm.com/docs/api/qiskit/qiskit.circuit.library.StatePreparation}{qiskit} & 1 & 4 & 11 & 26 & 57 & 120 & 247 & 502 & 1013 & 2036 & 4083 & 8178 & 16369 & 32752  \\

        LRSP~\cite{10.1109/TCAD.2023.3297972,qclib} & \href{https://github.com/qclib/qclib}{qclib} & 1 & 4 & 9 & 21 & 46 & 100 & 213 & 441 & 914 & 1862 & 3789 & 7650 & 15427 & 31000  \\
        PB + Block-ZXZ & \href{https://github.com/zexianLIPolyU/SPDMM-SIABLE}{spdmm} & 1 & 4 & 9 & 25 & 44 & 120 & 203 & 531 & 877 & 2237 & 3619 & 9155 & 14772 & 37108 \\
        LRSP + Block-ZXZ & \href{https://github.com/zexianLIPolyU/SPDMM-SIABLE}{spdmm} & 1 & 4 & 9 & 21 & 44 & 99 & 203 & 436 & 872 & 1841 & 3619 & 7565 & 14745 & 30659 \\
        \makecell[l]{\textbf{SPDMM}\\ \textbf{(Proposed method)}} & \href{https://github.com/zexianLIPolyU/SPDMM-SIABLE}{spdmm} & 1 & \textbf{3} & \textbf{7} & \textbf{18} & \textbf{42} & \textbf{93} &\textbf{199} & \textbf{418} & \textbf{867} & \textbf{1774} & \textbf{3612} & \textbf{7303} & \textbf{14736} & \textbf{29627}  \\ \hline
        \makecell[l]{Lower bound on\\ state preparation (Thm.~\ref{theorem: lower bound of C-NOT for state preparation})} & & 1 & 2 & 5 & 12 & 28 & 59 & 122 & 249 & 505 & 1016 & 2039 & 4086 & 8182 & 16373 \\
		\hline
	\end{tabular}
\end{table*}

\subsection{Contribution}
In this article, we optimize circuit size for both state preparation and block encoding. Both questions are optimized from the circuit structure, and improved by diagonal matrix migration, which is inspired by the Block-ZXZ method\mbox{~\cite{PhysRevApplied.22.034019}}. We generalize this idea to optimize the isometry methods and explore the mathematical structure of these two circuit synthesis questions. The C-NOT advantage is derived from the following parts:
\begin{itemize}
    \item The Block-ZXZ unitary synthesis method, which is proposed by Krol et al.~\cite{PhysRevApplied.22.034019}. We derive and demonstrate the diagonal migration idea building on this method.
    \item The new isometry synthesis method, which is inspired by the Block-ZXZ method.
    \item The new circuit synthesis structure for block encoding, which is derived from the singular value decomposition. This result reduces the C-NOT upper bound for block encoding by over 50\% relative to current unitary synthesis methods~\cite{PhysRevA.69.062321}. The latter was previously considered an effective method for block encoding.
\end{itemize}
The improvements of this article are elaborated in the following three dimensions.

\subsubsection{State Preparation Optimization}
 Our first main result improves the upper bound of state preparation to
\begin{equation}
    \begin{aligned}
    N_{{\rm state}}(n) &= \frac{11}{12} \sum_{i=0}^{l-3} 2^{\lfloor n / 2^i \rfloor} - \sum_{i=0}^{l-3} c_i 2^{\lfloor n / 2^{i+1} \rfloor}\\
    &\quad + (n - {\rm popcount}(n) - 1) + \sum_{i=0}^{l-3} d_i + \frac{1}{2} b (b - 1) ,
    \end{aligned}
    \label{eq C-NOT state preparation}
\end{equation}
where $l = \lfloor \log_2 n \rfloor + 1$ is the bit length of $n$;
$b = \lfloor n / 2^{l-2} \rfloor$ is the value of the two most significant bits of the integer $n$;
$c_i = 3$ if $\lfloor n / 2^i \rfloor$ is even; otherwise, $c_i = 4$;
$d_i = \frac{4}{3}$ if $\lfloor n / 2^i \rfloor$ is even; otherwise $d_i = \frac{5}{3}$; ${\rm popcount}(n)$ is the number of $1$ bits in the binary representation of $n$. The result is shown in Eq.~\eqref{eq C-NOT state preparation} and the formula will be demonstrated in Section~\ref{section state preparation}. Furthermore, we demonstrate the superiority of the SPDMM method in terms of weighted infidelity that accounts for both single-qubit and two-qubit gate errors.

\subsubsection{Single Ancilla Block Encoding Protocol with the Optimal Normalization Factor}
Our second main result improves the C-NOT count for synthesizing a single ancilla block encoding for a full-rank matrix with an optimal scaling factor:
\begin{theorem}
    Given a $2^{n-1}\times 2^{n-1}$ full-rank matrix, the single ancilla block encoding with the optimal normalization factor can be implemented using at most $\frac{11}{48} \times 4^n - 2^n + \frac{7}{3}$ C-NOT gates.
    \label{theorem block encoding}
\end{theorem}

\subsubsection{Single Ancilla Block Encoding Protocol for Low-rank Matrices}
Our third main result refines the leading term of the C-NOT expression for a rank-$K$ matrix block encoding to $(2^{\lceil\log_{2}K\rceil}+\frac{11}{12})2^n$, as shown in Theorem~\ref{theorem:low_rank_siable}.

To the best of our knowledge, all of these results currently represent the best-known upper bounds for C-NOT counts in the context of quantum circuit synthesis.

\subsection{Organization}
This article is structured as follows.
Section~\ref{section_Introduction} discusses the motivation and results of this work. Section~\ref{section_prelimiaries} presents some definitions and notations. Section~\ref{section state preparation} presents the implementation of our proposed state preparation algorithm. Section~\ref{section_siable} presents the single ancilla block encoding protocol with the optimal normalization factor. Section~\ref{section_experiments} provides numerical results for circuit compilation. Section~\ref{section_conclusion} offers concluding remarks.


\section{Preliminaries}
\label{section_prelimiaries}

\subsection{Definition}
When the $l_2$-norm of the vector $[\psi_k]_{k=0}^{2^n-1}\in\mathbb{C}^{2^n}$ is $1$, the quantum state preparation problem is to generate an $n$-qubit quantum state $\ket{\psi} = \sum_{k=0}^{2^n-1}\psi_k\ket{k}$ from $\ket{0}^{\otimes n}$, where $\{\ket{k}:k=0,1,\cdots,2^n-1\}$ is the computational basis.

\begin{definition}~\cite{10.1145/3313276.3316366} Suppose that $A$ is an $n$-qubit operator, $\alpha,\varepsilon\in \mathbb{R}+$, and $a\in \mathbb{N}$ with $m = a + n$. Then an $m$-qubit unitary $U_A$ is the $(\alpha, a,\varepsilon)$-block-encoding of $A$, if
$$
    \left\Vert A - \alpha\left(\bra{0}^{\otimes a} \otimes I_n\right) U \left(\ket{0}^{\otimes a} \otimes I_n\right) \right\Vert \leq \varepsilon.
$$
\end{definition}
The parameters $(\alpha, a, \varepsilon)$ represent respectively, normalization factor, number of ancilla qubits, and precision parameter. Ignoring precision simplifies it to $(\alpha, a)$-block-encoding.

\subsection{Notation}
The quantum multiplexor~\cite{1629135} can be represented as a block diagonal matrix, where each block is a unitary matrix. In the matrix form, a multiplexor with a single select bit is:
$$
\begin{bmatrix} U_0 & \\ & U_1 \end{bmatrix} = (U_0\oplus U_1),
$$
which can also be represented in the left circuit below. In particular, if $U_0$ is an identity matrix, it can be represented in the right.
\[
\Qcircuit @C=0.5em @R=0.5em {
& \qw & \ctrlsq{1} & \qw  \\
& {/}\qw & \gate{U} & \qw \\
}
\qquad \qquad
\Qcircuit @C=0.5em @R=0.5em {
& \qw & \ctrl{1} & \qw  \\
& {/}\qw & \gate{U_1} & \qw \\
}
\]

The uniformly controlled rotation, $\sum_{j=0}^{2^k-1}\ket{j}\bra{j}\otimes R_{\bm{a}}(\alpha_{j})$, consists of $k$-fold controlled rotations, and is a case of a quantum multiplexor with $k$ select bits. Given a unit vector $\bm{a} = (a_x,a_y,a_z)$, the special unitary operator is a one-parameter gate defined as $R_{\bm{a}}(\phi) = e^{i\bm{a}\cdot \bm{\sigma}\phi/2} = I\cos\frac{\phi}{2} + i(\bm{a}\cdot \bm{\sigma})\sin\frac{\phi}{2}$, where $I = \begin{bmatrix}
            1 & 0\\ 0 & 1
        \end{bmatrix}$; the product $\bm{\alpha}\cdot\bm{\sigma} = a_x X + a_y Y + a_z Z$ involves $X = \begin{bmatrix}
            0 & 1\\ 1 & 0
        \end{bmatrix}$, $Y = \begin{bmatrix}
            0 & -i\\ i & 0
        \end{bmatrix}$, $Z = \begin{bmatrix}
            1 & 0\\ 0 & -1
        \end{bmatrix}$. The Hadamard gate is denoted as $H=\frac{1}{\sqrt{2}}\begin{bmatrix} 1 & 1\\ 1 & -1 \end{bmatrix}$.

    \begin{lemma}
        The topology of an $n$-qubit circuit  utilizing $k$ C-NOT gates contains at most $3n+4k$ one-parameter gates.
        \label{lemma: one-parameter gates}
    \end{lemma}
    The demonstration of the above lemma follows the proof of proposition 1 in Ref.~\cite{PhysRevA.69.062321}, which illustrates the degree of freedom in an $n$-qubit quantum circuit with $k$ C-NOT gates.


\section{State Preparation via Diagonal Matrix Migration}
\label{section state preparation}
        Among the various approaches to state preparation circuit optimization—such as variational quantum state preparation\mbox{~\cite{PhysRevA.110.052615}} and quantum multiplexer–based decompositions\mbox{~\cite{10.1145/3748260}}—one influential line of work builds on the Schmidt decomposition of a pure state introduced by Plesch and Brukner (PB)\mbox{~\cite{PhysRevA.83.032302}}, which is the framework adopted in this paper.
 In the Schmidt decomposition, a quantum state is expressed as
$$
    \ket{\psi} = \sum_{i=0}^{2^{\lfloor n/2 \rfloor}} \sigma_i \ket{u_i}\ket{v_i},
$$
where $\ket{u_i}\in \mathbb{C}^{2^{\lceil n/2 \rceil}}$ and $\ket{v_i}\in \mathbb{C}^{2^{\lfloor n/2 \rfloor}}$, and $\left\{\ket{u_i}\right\}_i$ and $\left\{\ket{v_i}\right\}_i$ are two orthogonal bases. The PB method consist of four phases. In Phase 1, a superposition state whose amplitudes are the Schmidt coefficients is prepared; In Phase 2, $\lfloor n/2\rfloor$ C-NOT gates are implemented to prepare the state $\sum_{i=0}^{2^{\lfloor n/2\rfloor}-1} \sigma_i \ket{i} \otimes \ket{i}$; In Phase 3 and Phase 4, two orthogonal bases, $\left\{\ket{u_i}\right\}_i$ and $\left\{\ket{v_i}\right\}_i$, are prepared based on the state $\ket{i}$ via unitary synthesis. The complete process is as follows.
$$
\begin{aligned}
    &\ket{0}^{\otimes n} \xrightarrow{{\rm Phase}\ 1} \left(\sum_{i=0}^{2^{\lfloor n/2\rfloor}-1} \sigma_i \ket{i} \right) \otimes \ket{0}^{\otimes \lfloor n/2\rfloor} \\
    &\xrightarrow{{\rm Phase}\ 2} \sum_{i=0}^{2^{\lfloor n/2\rfloor}-1} \sigma_i \ket{i} \otimes \ket{i}
    \xrightarrow{{\rm Phase}\ 3,4} \sum_{i=0}^{2^{\lfloor n/2\rfloor}-1} \sigma_i \ket{u_i} \otimes \ket{v_i}. \\
\end{aligned}
$$

Prior work~\cite{PhysRevA.83.032302} merely identified the leading term of the bound on the C-NOT for the state preparation; the following theorem establishes the full polynomial bound on the C-NOT for the state preparation.

\begin{figure}
    \centering
    \includegraphics[width=1\linewidth]{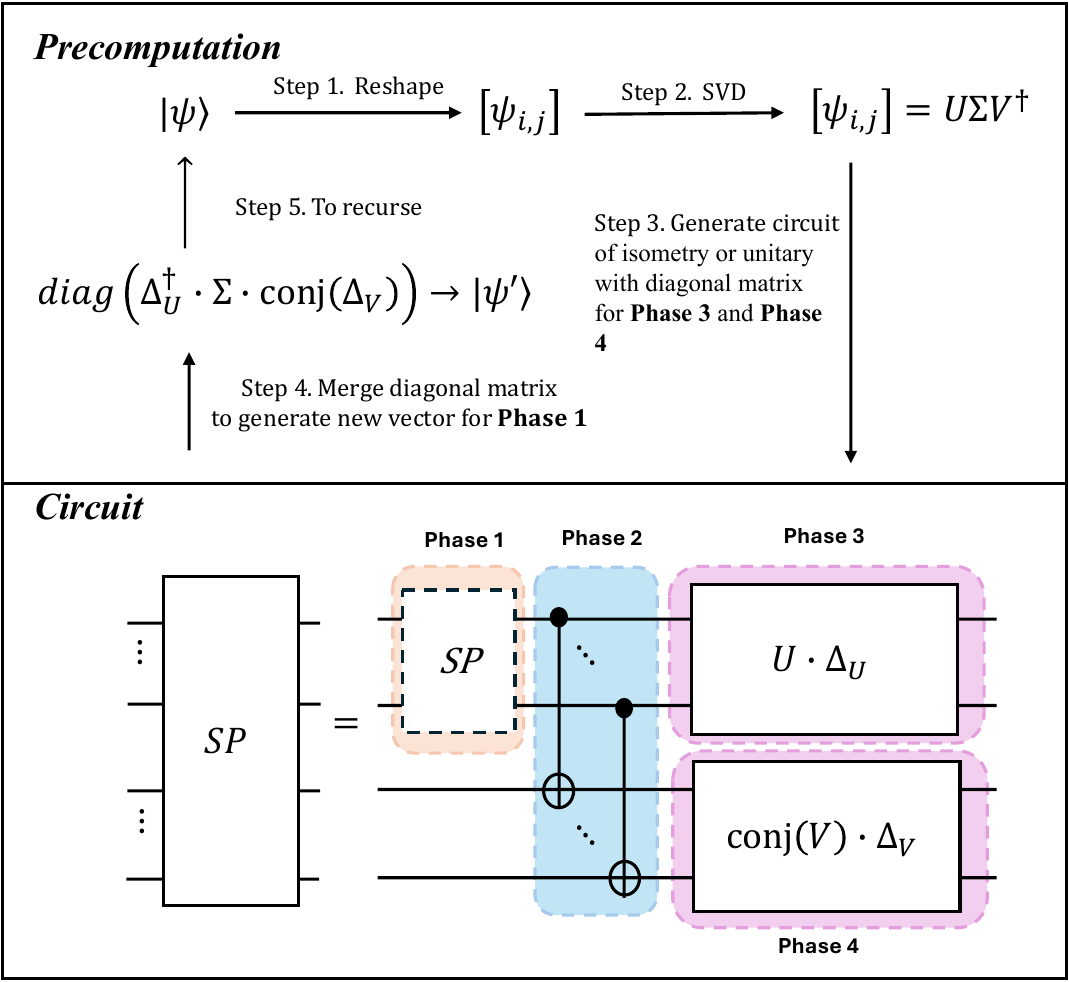}
    \caption{The process of synthesizing SPDMM. In the figure above, the green area is dedicated to processing the data, while the blue area is used for generating the quantum circuit. The algorithm breaks down the state preparation routine in each recursion into two subroutines: (1) Unitary/isometry synthesis subroutines generate $U$ and $\operatorname{conj}(V)$, up to diagonal matrices $\Delta_U$ and $\Delta_V$; (2) Another level recursion of state preparation subroutine prepares $\ket{\psi'} = \text{diag}(\Delta_U^\dagger \cdot \Sigma \cdot \operatorname{conj}(\Delta_V))$. The state preparation circuit can be recursively generated using the same process until $n=1$.}
    \label{fig:state_preparation}
\end{figure}


\begin{theorem}
\label{theorem: lower bound of C-NOT for state preparation}
    Consider any gate library consisting solely of C-NOT and single-qubit gates. In this framework, nearly all $n$-qubit quantum states cannot be prepared, unless the circuit employs at least $\lceil (1/2)\times2^n - (3/4)\times n - (1/4)\rceil$ C-NOTs.
\end{theorem}

\begin{proof}
    The proof follows a similar logic to that provided in~\cite{PhysRevA.69.062321}. By lemma~\ref{lemma: one-parameter gates}, a circuit topology utilizing $k$ C-NOT gates contains at most $3n+4k$ one-parameter gates. A general $n$-qubit state is fully described by $2^{n+1}-1$ real parameters. Consequently, it follows that $3n+4k \ge 2^{n+1}-1$, leading to $k \ge \lceil \frac{1}{2} \times 2^n - \frac{3}{4} \times n - \frac{1}{4} \rceil$.
\end{proof}

Our state preparation method features two key improvements. First, in Phases 3 and 4, synthesizing a unitary or isometry up to a diagonal matrix requires fewer C-NOT gates than synthesizing a full unitary/isometry. Second, the state preparation circuit for Phase 1 can be generated recursively and its resulting diagonal matrices can be incorporated into the next recursion. The state preparation method via diagonal matrix migration (SPDMM) proceeds as follows:
\begin{itemize}
    \item \textbf{Step 0.} Determine the qubit number $n$ of the state $\ket{\psi}$. When $n=1$, generate a single-qubit gate whose first column is $\ket{\psi} = \psi_0\ket{0} + \psi_1\ket{1}$, and return;
    \item \textbf{Step 1.} Reshape vector $[\psi_k]\in\mathbb{C}^{2^n}$ into a matrix $[\psi_{i,j}]\in\mathbb{C}^{2^{\lfloor n/2\rfloor}\times 2^{\lceil n/2\rceil}}$ by rearranging the indices as $\psi_k = \psi_{i,j}$ for $k = 2^{\lfloor n/2\rfloor}\times i + j$;
    \item \textbf{Step 2.} Apply singular value decomposition (SVD) to the matrix $\sum_{i,j}\psi_{i,j}\ket{i}\bra{j}= \sum_{i=0}^{2^{\lfloor n/2\rfloor}-1}\sigma_i u_i v_i^\dagger= U\Sigma V^\dagger $.
    \item \textbf{Step 3.} Generate circuits of $U$ and $\operatorname{conj}(V)$ (the complex conjugate of $V$) up to diagonal matrices $\Delta_U$ and $\Delta_{V}$ for Phases 3 and 4; Generate $\lfloor n/2\rfloor$ C-NOTs for Phase 2;
    \item \textbf{Step 4.} Merge the diagonal matrices $\Delta_U$, $\Sigma$, and $\operatorname{conj}(\Delta_{V})$ to generate $\ket{\psi'} = \text{diag}(\Delta_U^\dagger\cdot\Sigma \cdot \operatorname{conj}(\Delta_{V}))$;
    \item \textbf{Step 5.} Recursively apply the state generation process from Step 0 to prepare the state $\ket{\psi'} = \sum_{i=0}^{2^{\lfloor n/2\rfloor}-1}\tilde{\sigma}_i\ket{i}$ for Phase 1.
\end{itemize}

Steps 1 and 2 comprise a Schmidt decomposition. The key ingredients of the above process are $U\cdot\Delta_U$ and conj$(V)\cdot\Delta_V$, which can be regarded as $N\times (N-1)$ isometries, because $\Delta_U$ and $\Delta_V$ are arbitrary. In our protocol, $U\cdot\Delta_U$ and conj$(V)\cdot\Delta_V$ are realized by Block-ZXZ decomposition and optimized through diagonal matrix migration.

In the following, we first review the Block-ZXZ decomposition method in~\cite{PhysRevA.94.052317,PhysRevApplied.22.034019} for a unitary in subsection~\ref{subsection: Unitary Synthesis}. Then, we show in subsection~\ref{subsection: Isometry Synthesis} how this method can be generalized to isometry synthesis. In subsection~\ref{subsection:Matrix Migration Protocol}, we then explain how unitary and isometry synthesis can be simplified by diagonal matrix migration. Subsequently, we show how this improved isometry synthesis leads to more efficient quantum state preparation in subsection~\ref{subsection: Gate Counts} and analyze the precomputation cost in subsection~\ref{subsection: Pre-computation Analysis}.






\subsection{Unitary Synthesis}
\label{subsection: Unitary Synthesis}

\begin{figure*}[b]
    \centering
    \begin{subfigure}[b]{0.99\textwidth}
    \centering
    $
    \begin{myqcircuit*}{0.5}{0.2}
            & & & \\
            & & & \\
    		& \qw & \gate{R_Z} & \qw \\
            & \qw & \ctrlsq{-1} & \qw \\
    		& \qw & \ctrlsq{-1} & \qw \\
            & &
    \end{myqcircuit*}
    =
    \begin{myqcircuit*}{0}{0.2}
            & & & & & \mbox{$R^L$} & & & & & & & \\
    		& \qw & \gate{R_Z^{\tilde{\theta}_0}} & \targ & \gate{R_Z^{\tilde{\theta}_1}} & \targ & \gate{R_Z^{\tilde{\theta}_2}} & \targ & \gate{R_Z^{\tilde{\theta}_3}} & \qw & \qw & \targ & \qw \\
            & \qw & \qw & \qw & \qw & \ctrl{-1} & \qw & \qw & \qw & \qw & \qw & \ctrl{-1} & \qw \\
    		& \qw & \qw & \ctrl{-2} & \qw & \qw & \qw & \ctrl{-2} & \qw & \qw & \qw & \qw & \qw \gategroup{2}{3}{4}{9}{.7em}{--} \\
    \end{myqcircuit*}
    =
    \begin{myqcircuit*}{0}{0.2}
            & & & & & & & & \mbox{$R^R$} & & & & & & & & & &  \\
    		& \qw & \targ & \qw & \qw & \gate{R_Z^{\tilde{\theta}_3}} & \targ & \gate{R_Z^{\tilde{\theta}_2}} & \targ & \gate{R_Z^{\tilde{\theta}_1}} & \targ & \gate{R_Z^{\tilde{\theta}_0}} & \qw \\
    	    & \qw & \ctrl{-1} & \qw & \qw & \qw & \qw & \qw & \ctrl{-1} & \qw & \qw & \qw & \qw \\
            & \qw & \qw & \qw & \qw & \qw & \ctrl{-2} & \qw & \qw & \qw & \ctrl{-2} & \qw & \qw  \gategroup{2}{6}{4}{12}{.7em}{--} \\
    \end{myqcircuit*}
    $
    \caption{Circuit for decoupling uniformly controlled rotations about the $z$ axis (URCZ) $\sum_{k=0}^{3}R_Z(-2\theta_k)\otimes\ket{k}\bra{k} = R^L [(I\otimes\ket{0}\bra{0} +X\otimes\ket{1}\bra{1})\otimes I] = [(I\otimes\ket{0}\bra{0} +X\otimes\ket{1}\bra{1})\otimes I] R^L $. In the above notation, $R_Z^{\tilde{\theta}_k} = R_Z(\tilde{\theta}_k)$.  $\tilde{\theta}_k$ can be computed by $(H^{\otimes k} P_G)(\tilde{\theta}_0,\cdots,\tilde{\theta}_{2^k-1})^T = (\theta_0,\cdots,\theta_{2^k-1})^T$, where $H^{\otimes k}$ is the Walsh-Hadamard transformation~\cite{PhysRevLett.93.130502,1674569} and $P_G$ is the permutation matrix that transforms binary ordering to Gray code ordering~\cite{9951292}. The $R^L$ and $R^R$ denote uniformly controlled rotations up to a top C-NOT gate under these two distinct alternating sequences.}
    \label{fig:uniformly controlled rotations}
    \end{subfigure}
    \hfill
    \vspace{10pt}
    \hfill
    \begin{subfigure}[b]{0.99\textwidth}
    \centering
    $
        \begin{myqcircuit*}{0.5}{0.2}
         & \qw & \qw & \multigate{2}{U} & \qw \\
         & \qw & \qw & \ghost{U} & \qw \\
         \ustick{n-2\qquad  } & \qw & {/}\qw & \ghost{U} & \qw \\
        \end{myqcircuit*}
        =
        \begin{myqcircuit*}{0.2}{0.5}
        		& \qw & \qw & \multigate{2}{R^L} & \targ & \gate{H} & \ctrl{1} & \gate{H} & \targ & \multigate{2}{R^R} & \qw & \qw \\
                & \qw & \multigate{1}{V_N^\dagger} & \ghost{R^L} & \ctrl{-1} & \multigate{1}{W_N} & \multigate{1}{L} & \multigate{1}{V_M^\dagger} & \ctrl{-1} & \ghost{R^R} & \multigate{1}{W_M} & \qw \\
                & {/} \qw & \ghost{V_N^\dagger} & \ghost{R^L} & \qw & \ghost{W_N} & \ghost{L} & \ghost{V_M^\dagger} & \qw & \ghost{R^R} & \ghost{W_M} & \qw \gategroup{1}{5}{3}{9}{0.95em}{--} \\
        \end{myqcircuit*}
        =
        \begin{myqcircuit*}{0.2}{0.2}
        		& \qw & \qw & \multigate{2}{R^L}& \gate{H} & \gate{R_Z} & \gate{H} & \multigate{2}{R^R} & \qw & \qw \\
                & \qw & \multigate{1}{V_N^\dagger} & \ghost{R^L} & \multigate{1}{V_{\tilde{L}}^\dagger} & \ctrlsq{-1} & \multigate{1}{W_{\tilde{L}}} & \ghost{R^R} & \multigate{1}{W_M} & \qw \\
        		& {/} \qw & \ghost{V_N^\dagger} & \ghost{R^L} & \ghost{V_{\tilde{L}}^\dagger} & \ctrlsq{-1} & \ghost{W_{\tilde{L}}} & \ghost{R^R} & \ghost{W_M} & \qw \gategroup{1}{5}{3}{7}{0.45em}{--} \\
        \end{myqcircuit*}
        $
          \caption{C-NOT reduction of unitary synthesis: Merging the adjoint two C-NOT gates in the decoupling circuit of URCZ, where $R^L$ and $R^R$ are defined in Fig.~\ref{fig:uniformly controlled rotations}. Two unitaries $V_N^\dagger,W_M$, a multiplexor operation $I_{2^{n-1}}\oplus L$, and two Hadamard gates are merged into a new multiplexor $\tilde{L}= \left(V_M^\dagger W_N\right) \oplus \left((Z\otimes I)V_M^\dagger L W_N(Z\otimes I)\right)$, $V_{\tilde{L}}^\dagger$ and $W_{\tilde{L}}$ are derived from decoupling $\tilde{L}$ using Eq.~\eqref{eq_quantum_multiplexor}.}
        \label{subfigure merge}
    \end{subfigure}
    \hfill
    \vspace{10pt}
    \hfill
    \begin{subfigure}[b]{0.99\textwidth}
    \centering
    $
    \begin{myqcircuit*}{0.5}{0.5}
     & \qw & \qw & \multigate{2}{U_{iso}} & \qw \\
     & \qw & \qw & \ghost{U_{Iso}} & \qw \\
     \ustick{n-2} & {/}\qw & \qw & \ghost{U_{Iso}} & \qw \\
    \end{myqcircuit*}
    =
    \begin{myqcircuit*}{0.2}{0.5}
     & \gate{H} & \ctrl{1} & \gate{H} & \ctrlsq{1} & \qw \\
     & \qw & \multigate{1}{L} & \qw & \multigate{1}{M} & \qw \\
     & {/}\qw & \ghost{L} & \qw & \ghost{M} & \qw \\
    \end{myqcircuit*}
    =
    \begin{myqcircuit*}{0.5}{0.55}
    		& \qw & \qw & \gate{H} & \ctrl{1} & \gate{H} & \targ & \multigate{2}{R^R} & \qw & \qw \\
            & \qw & \qw & \qw & \multigate{1}{L} & \multigate{1}{V_M^\dagger} & \ctrl{-1} & \ghost{R^R} & \multigate{1}{W_M} & \qw \\
            & {/} \qw & \qw & \qw & \ghost{L} & \ghost{V_M^\dagger} & \qw & \ghost{R^R} & \ghost{W_M} & \qw \gategroup{1}{4}{3}{7}{0.95em}{--} \\
    \end{myqcircuit*}
     =
    \begin{myqcircuit*}{0.2}{0.5}
    		& \qw & \qw  & \gate{H} & \gate{R_Z} & \gate{H} & \multigate{2}{R^R} & \qw & \qw \\
            & \qw & \qw & \multigate{1}{V_{\hat{L}}^\dagger} & \ctrlsq{-1} & \multigate{1}{W_{\hat{L}}} & \ghost{R^R} & \multigate{1}{W_M} & \qw \\
    		& {/} \qw & \qw & \ghost{V_{\hat{L}}^\dagger} & \ctrlsq{-1} & \ghost{W_{\hat{L}}} & \ghost{R^R} & \ghost{W_M} & \qw \gategroup{1}{4}{3}{6}{0.95em}{--}  \\
    \end{myqcircuit*}
    $
    \caption{C-NOT reduction of isometry synthesis: Merging the adjoint C-NOT gate in the decoupling circuit of URCZ, unitary $V_M^\dagger$, a multiplexor operation $I_{2^{n-1}}\oplus L$, and two Hadamard gates into a new multiplexor $\hat{L} = \left(V_M^\dagger\right) \oplus \left((Z\otimes I)V_M^\dagger L(Z\otimes I)\right)$, where $V_{\hat{L}}^\dagger$ and $W_{\hat{L}}$ are derived from decoupling $\hat{L}$ using Eq.~\eqref{eq_quantum_multiplexor}.}
    \label{fig:Isometry}
    \end{subfigure}
    \caption{Quantum circuit of unitary and isometry synthesis by the Block-ZXZ decomposition.}
\end{figure*}

We perform the unitary synthesis in the following three levels:


\subsubsection{Decoupling a $\mathrm{U}(2^{n})$ Operation into $n$-qubit Multiplexor Operations with $n\ge 3$}

Any unitary $U\in \mathrm{U}(2^{n})$ with $n\ge 3$ can be partitioned into four parts as follows:
\begin{equation}
    \begin{aligned}
    U &= \begin{bmatrix}
        A & B \\ C & D
        \end{bmatrix}
    = \frac{1}{2} \begin{bmatrix}
        M_1(I+L) & M_1(I-L)N \\
        M_2(I-L) & M_2(I+L)N \\
    \end{bmatrix}\\
    &= \begin{bmatrix}
    M_1 & \\
     & M_2 \\
    \end{bmatrix} \left(H\otimes I\right) \begin{bmatrix}
    I & \\
     & L \\
    \end{bmatrix} \left(H\otimes I\right) \begin{bmatrix}
    I & \\
     & N \\
    \end{bmatrix},
    \end{aligned}
    \label{eq: Transformed Block-ZXZ}
\end{equation}
 where $I$ is the identity matrix, $M_1 = (S_A + iS_B)U_A, M_2= C + iDU_B^\dagger U_A, N^\dagger = iU_B^\dagger U_A, L = 2M_1^\dagger A - I \in \mathrm{U}(2^{n-1})$. Here, $S_A$ and $S_B$ are symmetric positive semidefinite matrices, and $U_A$ and $U_B$ are the unitary matrices obtained from the polar decompositions $A = S_A U_A$ and $B = S_B U_B$~\cite{doi:10.1137/1.9781421407944}. The above process decouples $U$ into a product of three multiplexors $(M_1 \oplus M_2)$, $(I \oplus L)$, $(I\oplus N)$ and two Hadamard gates with classical precomputation of $\mathcal{O}(2^{3n/2})$.

\subsubsection{Decoupling an $n$-qubit Multiplexor into $\mathrm{U}(2^{n-1})$ Gates and a Uniformly Controlled Rotation with $n\ge 2$}

A quantum multiplexor in $\mathbb{C}^{2^n\times 2^n}$ can be decomposed as~\cite{1629135}
\begin{equation}
    A_1 \oplus A_2
     =
    (I_{2}\otimes W_{A})
    (D \oplus D^\dagger) (I_{2}\otimes V_{A}^\dagger),
    \label{eq_quantum_multiplexor}
\end{equation}
where $W_{A}$, $V_{A}$ and $D$ are derived from the eigenvalue decomposition as $A_1A_2^\dagger = W_{A} D^2 W_{A}^\dagger$, and $V_{A}^\dagger = DW_{A}^\dagger A_2$.
$ D \oplus D^\dagger$ constitutes a uniformly controlled rotation about the $z$ axis (UCRZ), and
it can be implemented using the circuit shown
in Fig.~\ref{fig:uniformly controlled rotations}, which derives C-NOT counts as
$$
    N_{\text{UCRZ}}(n) = N_{R^L}(n) + 1 = N_{R^R}(n) + 1 = 2^{n-1}
$$
 In the circuit of unitary synthesis, two C-NOT gates from the decomposition of URCZ in the Block-ZXZ decomposition can be merged into a new multiplexor $\tilde{L}$, where the reduction process~\cite{PhysRevApplied.22.034019} is as follows,
 $$
 \begin{aligned}
       &(CX\otimes I)(H\otimes V_M^\dagger)(I\oplus L)(H\otimes W_N)(CX\otimes I) \\
       =& (H\otimes I)\left( V_M^\dagger W_M \oplus (Z\otimes I)V_M^\dagger L W_M(Z\otimes I) \right)(H\otimes I) \\
       =& (H\otimes I) \tilde{L}(H\otimes I) .
 \end{aligned}
 $$
 The whole reduction process is shown in Fig.~\ref{subfigure merge}, and the classical precomputation cost is $\mathcal{O}(2^{3(n-1)/2})$.

\subsubsection{Decoupling a $\mathrm{U}(4)$ Operation into Single-qubit Gates and C-NOT Gates Up to a Diagonal Matrix}
\label{subsection: U4}

Unitary matrices in the Lie group $\mathrm{U}(4)$ can be represented as the product of a global phase factor and an element of the special unitary group $\mathrm{SU}(4)$. Quantum circuits of $\mathrm{SU}(4)$ decomposition up to a diagonal matrix can be synthesized using $2$ C-NOT gates~\cite{PhysRevA.69.062321}, as shown below, where $a,b,c,d\in\mathrm{SU}(2)$ are four one-qubit gates.
$$
\begin{myqcircuit*}{0.5}{0.7}
		& \qw & \multigate{1}{U} & \multigate{1}{\Delta} & \qw \\
		& \qw & \ghost{\Delta} & \ghost{U}  & \qw \\
\end{myqcircuit*}
=
\begin{myqcircuit*}{0.}{0.7}
		& \qw & \gate{a} & \ctrl{1} & \gate{R_X(\theta)} & \ctrl{1} & \gate{c} & \qw \\
		& \qw & \gate{b} & \targ & \gate{R_Z(\phi)} & \targ & \gate{d} & \qw \\
\end{myqcircuit*}
$$
The single-qubit and two-qubit gates decomposition of a $\text{U}(4)$ operation is specified in appendix~\ref{appendix: Decoupling two-qubit operations}. In subsection~\ref{subsection:Matrix Migration Protocol}, we present a method for migrating diagonal matrices, with the goal of reducing the overall synthesis cost.

\subsection{Block-ZXZ based Isometry Synthesis}
\label{subsection: Isometry Synthesis}

The Block-ZXZ decomposition of a unitary, as stated in Eq.~\eqref{eq: Transformed Block-ZXZ}, synthesizes the matrix using three multiplexor operations. In contrast to this conventional approach for a full unitary, we consider the synthesis of an isometry from $\mathbb{C}^{2^{n-1}}$ to $\mathbb{C}^{2^{n}}$, corresponding to the first $2^{n-1}$ columns of an $n$-qubit unitary. For this task, we show that our novel method requires only two multiplexor operations, as follows.
\begin{equation*}
    \begin{aligned}
    U_{iso} = \begin{bmatrix}
        A & * \\
        C & * \\
    \end{bmatrix}
    &= \frac{1}{2}\begin{bmatrix}
        M_1 & \\
         & M_2
    \end{bmatrix}\begin{bmatrix}
        I+L & I-L \\
        I-L & I+L
    \end{bmatrix} \\
    &= \begin{bmatrix}
        M_1 & \\
         & M_2 \\
        \end{bmatrix} \left(H\otimes I\right) \begin{bmatrix}
        I & \\
         & L \\
        \end{bmatrix} \left(H\otimes I\right).
    \end{aligned}
\end{equation*}
To further optimize C-NOT gates, one C-NOT can
be merged into a new multiplexor $\hat{L}$ in the isometry synthesis, where the reduction process is as follows,
$$
\begin{aligned}
    &(CX\otimes I)(H\otimes V_M^\dagger)(I\oplus L)(CX\otimes I) \\
    = & (H\otimes I) \left(V_M^\dagger\oplus (Z\otimes I)V_M^\dagger L(Z\otimes I)\right) (H\otimes I) \\
    = & (H\otimes I) \hat{L} (H\otimes I).
\end{aligned}
$$
The whole process is illustrated in Fig.~\ref{fig:Isometry}.


\subsection{Circuit Synthesis via Diagonal Matrix Migration Protocol}
\label{subsection:Matrix Migration Protocol}

Both unitary and isometry synthesis can be decomposed into C-NOT gates and $\text{U}(4)$ operations. To further reduce the implementation cost, we exploit the commutativity of diagonal matrices decoupled from each $\text{U}(4)$ operation. These diagonal operators can be commuted through both the topmost C-NOT gate and any uniformly controlled $z$-rotation gates (URCZ). Consequently, such a diagonal matrix can also be migrated through the $R^L$ operation ($R^L$ is defined in Fig.~\ref{fig:uniformly controlled rotations}).

$$
\begin{myqcircuit*}{0.4}{0.7}
		& \qw & \qw & \multigate{2}{R^L} & \qw \\
        & \qw & \multigate{1}{\Delta} & \ghost{R^L} & \qw \\
	    & {/} \qw & \ghost{\Delta} & \ghost{R^L} & \qw \\
\end{myqcircuit*}
 =
\begin{myqcircuit*}{0.}{0.7}
		& \qw & \qw & \gate{R_Z} & \targ & \qw \\
        & \qw & \multigate{1}{\Delta} & \ctrlsq{-1} & \ctrl{-1} & \qw \\
		  & {/} \qw & \ghost{\Delta} & \ctrlsq{-1} & \qw & \qw \\
\end{myqcircuit*}
 =
\begin{myqcircuit*}{0.4}{0.7}
		& \qw & \multigate{2}{R^L} & \qw & \qw \\
        & \qw & \ghost{R^L} & \multigate{1}{\Delta} & \qw \\
		  & {/} \qw & \ghost{R^L} & \ghost{\Delta} & \qw \\
\end{myqcircuit*}
$$

Furthermore, the commutativity of the diagonal matrices $I\otimes\Delta$ and $R^L$ can be generalized to the commutativity of $I\otimes\Delta$ and $R^R$ (as defined in Fig.~\ref{fig:uniformly controlled rotations}).  Based on the above commutativity, the C-NOT counts in the decomposition of isometry/unitary synthesis up to a diagonal matrix can be reduced by migrating through $R^L$, $R^R$ and URCZ and merging into the next $\text{U}(4)$ operation in the recursive relationship.

Consider a $3$-qubit example. Since synthesizing a $2$-qubit unitary up to a diagonal matrix requires fewer C-NOT gates, we can synthesize $\Delta \cdot U_1$ instead of $U_1$ directly, then migrate $\Delta$ through the $R^L$ and absorb it into $U_2$ into a new $\text{U}(4)$ matrix $U_2\cdot \Delta$. This procedure is depicted below and is suitable for application in subsequent recursions in the unitary synthesis and isometry synthesis.
$$
\begin{myqcircuit*}{0.4}{0.2}
    & \qw & \multigate{2}{R^L} & \qw & \qw & \qw \\
    & \multigate{1}{U_1} & \ghost{R^L} & \multigate{1}{U_2} & \qw \\
    & \ghost{U_1} & \ghost{R^L} & \ghost{U_2} & \qw \\
\end{myqcircuit*}
=
\begin{myqcircuit*}{0.4}{0.2}
    & \qw & \multigate{2}{R^L} & \qw & \qw & \qw \\
    & \multigate{1}{\Delta\cdot U_1} & \ghost{R^L} & \multigate{1}{U_2\cdot\Delta} & \qw \\
    & \ghost{\Delta\cdot U_1} & \ghost{R^L} & \ghost{U_2\cdot\Delta} & \qw \\
\end{myqcircuit*}
$$


\subsection{Gate Counts in the State Preparation}
\label{subsection: Gate Counts}
Based on Fig.~\ref{subfigure merge}, the number of C-NOT gates for synthesizing an $n$-qubit unitary up to a diagonal matrix is
$$
\begin{aligned}
N_{u\cdot\Delta}(n) &\leq N_{R^L}(n) + N_{R^R}(n) + N_{\text{UCRZ}}(n) + 4N_{u\cdot\Delta}(n-1)\\
&\leq  4N_{u\cdot\Delta}(n-1) + 3 \times 2^{n-1} - 2.
\end{aligned}
$$
Based on Fig.~\ref{fig:Isometry}, the number of C-NOT gates for synthesizing a $\mathbb{C}^{2^{n-1}}$ to $\mathbb{C}^{2^{n}}$ isometry is
$$
\begin{aligned}
    N_{{\rm iso} \cdot\Delta}(2^{n-1},n) &\leq 3N_{u\cdot\Delta}(n-1) + N_{R^R}(n) + N_{\text{UCRZ}}(n), \\
    &\leq 3N_{u\cdot\Delta}(n-1) + 2^n-1.
\end{aligned}
$$
Since $N_{u\cdot \Delta}(2) = 2$~\cite{PhysRevA.69.062321}, it follows that for $n\ge 2$,
\begin{align*}
    & N_{u\cdot\Delta}(n) \leq \frac{11}{24} \times 4^n - \frac{3}{2}\times 2^n + \frac{2}{3}, \\
    & N_{{\rm iso}\cdot\Delta}(2^{n-1},n) \leq \frac{11}{32}\times 4^n - \frac{5}{4}\times 2^n + 1. \label{C-NOT isometry}
\end{align*}

Based on the SPDMM decomposition, the number of C-NOT gates for state preparation depends on the number of C-NOT gates for unitary/isometry up to a diagonal matrix as
$$
\begin{aligned}
    &N_{{\rm state}}(n) \leq \\
    &\ \left\{
    \begin{aligned}
    & N_{{\rm state}}(n/2) + 2N_{u\cdot \Delta}(n/2) + n/2, &&\text{$n$ is even}; \\
    & N_{{\rm state}}(\lfloor n/2\rfloor) +  N_{{\rm iso}\cdot\Delta}(2^{\lfloor n/2\rfloor}, \lceil n/2\rceil)\\
    &\qquad\qquad\qquad + N_{u\cdot\Delta}(\lfloor n/2\rfloor)+\lfloor n/2\rfloor,  &&\text{$n$ is odd}. \\
\end{aligned}
\right.
\end{aligned}
$$
This yields the expression for $N_{state}(n)$ in Eq.~\eqref{eq C-NOT state preparation}.

The data processing steps for SPDMM are shown in Fig.~\ref{fig:state_preparation}. For illustration purposes, the circuit for state preparation on $2-5$ qubits is shown in Fig.~\ref{fig:state_preparation_demo}. We also give a concrete example for 3-qubit state preparation in Appendix~\ref{appendix:example}.

\begin{figure*}[htbp]
    \centering
    \begin{subfigure}[htbp]{0.2\textwidth}
    \centering
        $$
        \begin{myqcircuit*}{0.2}{0.2}
            & \qw & \gate{U} & \ctrl{1} & \gate{U} & \qw \\
            & \qw & \qw & \targ & \gate{U} & \qw \\
        \end{myqcircuit*}
        $$
        \caption{$2$-qubit with $1$ C-NOT}
    \end{subfigure}
    \begin{subfigure}[htbp]{0.2\textwidth}
    \centering
        $$
        \begin{myqcircuit*}{0.}{0.3}
            & \qw & \qw & \qw & \gate{U} & \ctrl{1} & \gate{U} & \ctrl{1} & \gate{U} & \qw \\
            & \qw & \gate{U} & \ctrl{1} & \gate{U} & \targ & \gate{U} & \targ & \gate{U} & \qw \\
            & \qw & \qw & \targ & \gate{U} & \qw & \qw & \qw & \qw & \qw \\
        \end{myqcircuit*}
        $$
        \caption{$3$-qubit with $3$ C-NOTs}
    \end{subfigure}
    \hspace{5pt}
    \begin{subfigure}[htbp]{0.24\textwidth}
    \centering
        $$
        \begin{myqcircuit*}{0.}{0.2}
            & \gate{U} & \ctrl{1} & \gate{U} & \ctrl{2} & \qw & \gate{U} & \ctrl{1} & \gate{U} & \ctrl{1} & \gate{U} & \qw \\
            & \qw & \targ & \gate{U} & \qw & \ctrl{2} & \gate{U} & \targ & \gate{U} & \targ & \gate{U} & \qw \\
            & \qw & \qw & \qw & \targ & \qw & \gate{U} & \ctrl{1} & \gate{U} & \ctrl{1} & \gate{U} & \qw \\
            & \qw & \qw & \qw & \qw & \targ & \gate{U} & \targ & \gate{U} & \targ & \gate{U} & \qw \\
        \end{myqcircuit*}
        $$
        \caption{$4$-qubit with $7$ C-NOTs}
    \end{subfigure}
    \hfill
    \\
    \begin{subfigure}[htbp]{0.98\textwidth}
    \centering
        $$
        \begin{myqcircuit*}{0.2}{0.2}
            & \qw & \qw & \qw & \qw & \qw & \qw & \qw & \qw & \qw & \gate{U} & \targ & \gate{U} & \targ & \gate{U} & \targ & \gate{U} & \targ & \gate{U} & \qw & \qw & \qw & \gate{U} & \targ & \gate{U} & \targ & \gate{U} & \targ & \gate{U} & \qw & \qw & \qw & \qw & \qw  \\
            & \gate{U} & \ctrl{1} & \gate{U} & \ctrl{2} & \qw & \gate{U} & \ctrl{1} & \gate{U} & \ctrl{1} & \gate{U} & \ctrl{-1} & \qw & \qw & \qw & \ctrl{-1} & \qw & \qw & \gate{U} & \ctrl{1} & \gate{U} & \ctrl{1} & \gate{U} & \qw & \qw & \ctrl{-1} & \qw & \qw & \gate{U} & \ctrl{1} & \gate{U} & \ctrl{1} & \gate{U} & \qw \\
            & \qw & \targ & \gate{U} & \qw & \ctrl{2} & \gate{U} & \targ & \gate{U} & \targ & \gate{U} & \qw & \qw & \ctrl{-2} & \qw & \qw & \qw & \ctrl{-2} & \gate{U} & \targ & \gate{U} & \targ & \gate{U} & \ctrl{-2} & \qw & \qw & \qw & \ctrl{-2} & \gate{U} & \targ & \gate{U} & \targ & \gate{U} & \qw  \\
            & \qw & \qw & \qw & \targ & \qw & \gate{U} & \ctrl{1} & \gate{U} & \ctrl{1} & \gate{U} & \qw & \qw & \qw & \qw & \qw & \qw & \qw & \qw & \qw & \qw & \qw & \qw & \qw & \qw & \qw & \qw & \qw & \qw & \qw & \qw & \qw & \qw & \qw \\
            & \qw & \qw & \qw & \qw & \targ & \gate{U} & \targ & \gate{U} & \targ & \gate{U} & \qw & \qw & \qw & \qw & \qw & \qw & \qw & \qw & \qw & \qw & \qw & \qw & \qw & \qw & \qw & \qw & \qw & \qw & \qw & \qw & \qw & \qw & \qw \\
        \end{myqcircuit*}
        $$
        \caption{$5$-qubit with $18$ C-NOTs}
    \end{subfigure}
    \caption{The circuit for preparing $2$-$5$ qubit states via SPDMM.}
    \label{fig:state_preparation_demo}
\end{figure*}

\subsection{Classical Pre-computation Analysis}
\label{subsection: Pre-computation Analysis}
    The primary computational cost of the SPDMM decomposition mainly arises from three key components:
    \begin{itemize}
    \item The Schmidt decomposition (steps 1 and 2) at the top level of the recursion, which reduces an $n$-qubit state preparation problem to two $(n/2)$-qubit unitary (or isometry) synthesis problems and one additional $(n/2)$-qubit state preparation problem, incurs a classical precomputation cost of $\mathcal{O}(2^{3n/2})$.
    \item The $(n/2)$-qubit unitary or isometry synthesis problem (step 3) requires a classical precomputation cost of $\mathcal{O}(2^{3n/2})$ in total, dominated by the top level.
    \item Diagonal matrix migration (step 4) incurs a classical precomputation cost of $\mathcal{O}(2^{n/2})$ per level.
    \end{itemize}
    Consequently, the total classical precomputation cost of SPDMM is $\mathcal{O}(2^{3n/2})$.

\section{Single ancilla block encoding}
\label{section_siable}

In this section, single ancilla block encoding (SIABLE) for full-rank and low-rank matrices will be introduced.

\subsection{Full Rank Matrix Encoding}
\label{subsection full rank encoding}
For a full-rank matrix $A\in \mathbb{C}^{2^{n-1} \times 2^{n-1}}$ with $\Vert A\Vert_2 \leq 1$, it can be diagonalized as $A = W_A\Sigma_A V_A^\dagger$, where $\Sigma_A$ is a positive definite diagonal matrix. Denote $\Sigma_A = \text{diag}(\cos(\theta_1),\cos(\theta_2),\cdots,\cos(\theta_{2^{n-1}}))$, then there are two unitaries $A_1$ and $A_2$ such that $A = \frac{A_1 + A_2}{2}$, where
$$
\begin{aligned}
    A_1 &= W_A\text{diag}(e^{i\theta_1}, \cdots, e^{i\theta_{2^{n-1}}}) V_A^\dagger, \\
    A_2 &= W_A\text{diag}(e^{-i\theta_1},\cdots, e^{-i\theta_{2^{n-1}}}) V_A^\dagger.
\end{aligned}
$$
Inthis way, a block encoding of $A$ can be synthesized by
\begin{equation}
    U_A = (H\otimes I_{2^{n-1}} ) (A_1\oplus A_2) (H\otimes I_{2^{n-1}} ) .
\label{eq multiplexor decomposition}
\end{equation}

\begin{figure}[htbp]
    \centering
    $$
    \begin{myqcircuit*}{1.2}{0.5}
        & \qw & \gate{H} & \ctrlsq{1} & \gate{H} & \qw \\
        \ustick{n-1} & {/} \qw & \qw & \gate{A} & \qw & \qw \\
    \end{myqcircuit*}
    =
    \begin{myqcircuit*}{0.3}{0.5}
        & \qw & \gate{H} & \gate{R_Z} & \gate{H} & \qw \\
        & {/} \qw & \gate{\Delta_{V_A}\cdot V_A^\dagger} & \ctrlsq{-1} & \gate{W_A\cdot\Delta_{V_A}^\dagger} & \qw \\
    \end{myqcircuit*}
    $$
    \caption{Quantum circuit of single ancilla block encoding for a matrix in $\mathbb{C}^{2^{n-1}\times 2^{n-1}}$ with spectral norm $\Vert A\Vert_2\leq 1$.}
    \label{fig:siable}
\end{figure}
The quantum circuit of the block encoding protocol shown in Fig.~\ref{fig:siable} consists of two unitary matrices, $W_A$ and $V_A^\dagger$, along with an $(n-1)$-fold uniformly controlled rotation about the $z$ axis. The decomposition of the unitary $V_A^\dagger$ can be optimized by reducing one C-NOT gate through migrating a diagonal matrix $\Delta_{V_A}$ to $W_A$. Denote the number of C-NOT gates for synthesizing a block encoding of a $2^{n-1}\times 2^{n-1}$ matrix with spectral norm $\Vert A\Vert_2\leq 1$ as $N_{\mathrm{BE}}(n-1)$.
Since the unitary synthesis takes the C-NOTs $N_{u}(n) \leq \frac{11}{24} \times 4^n - \frac{3}{2}\times 2^n + \frac{5}{3}$~\cite{PhysRevApplied.22.034019}, which derive the result of Thm.~\ref{theorem block encoding} as
$$
    N_{{\text{BE}}}(n-1) \leq \frac{11}{48} \times 4^n -  2^n + \frac{7}{3}.
$$
Note that the upper bound of C-NOT gates in the single ancilla block encoding protocol is lower than the lower bound for $n$-qubit unitary synthesis given by $\lceil \frac{1}{4}(4^n - 3n -1)\rceil$~\cite{PhysRevA.69.062321}. The following theorem generalizes the lower bound of the C-NOT number of block encoding for full-rank matrices with a single ancilla.

\begin{theorem}
\label{theorem: lower bound of C-NOT for block encoding}
    Consider any gate library consisting solely of C-NOT and single-qubit gates. In this framework, nearly all $2^{n-1}\times2^{n-1}$ matrices cannot be block-encoded with $\Vert A\Vert_2$ as the normalization factor with a single ancilla, unless the circuit employs at least $\lceil (1/8)\times4^n - (3/4)\times n \rceil$ C-NOTs.
\end{theorem}
\begin{proof}
    By Lemma~\ref{lemma: one-parameter gates}, an $n$-qubit circuit with $k$ C-NOTs has up to $3n+4k$ one-parameter gates.
    A $2^{n-1} \times 2^{n-1}$ matrix has $4^{n-1}-1$ degrees of freedom.
    Consequently, it follows that $3n+4k \ge 4^{n-1}$ and leads to $k \ge \lceil \frac{1}{8} \times 4^n - \frac{3}{4} \times n \rceil$.
\end{proof}

Table~\ref{tab: SIABLE} provides the results for single-ancilla block encoding applied to full-rank matrices.



\begin{table*}[htbp]
    \centering
    \caption{Comparison of the numbers of C-NOT gates and normalization factor $\alpha$ between the single ancilla block encoding protocol (SIABLE) for general $2^{n-1}\times 2^{n-1}$ full-rank complex matrix and other block encoding protocols and bounds in an $n$-qubit system. The numbers in the first row refer to $n$, where $n=3,4,5,6,7$.}
    \label{tab: SIABLE}
    \begin{tabular}{|l|l|c|cccccc|}
        \hline
        Methods & Script & $\alpha$ & 3 & 4 & 5 & 6 & 7 & \textbf{$n$} \\
        \hline
        FABLE~\cite{9951292} &\href{https://github.com/QuantumComputingLab/fable}{fable} & $2^{n-1}\Vert A\Vert_{\infty}$  & 32 & 128 & 512 & 2048 & 8192 & $(1/2)\times4^n$ \\
        BITBLE~\cite{11216010} & \href{https://github.com/zexianLIPolyU/BITBLE}{bitble} & $\Vert A\Vert_F$ & 30 & 126 & 510 & 2046 & 8190 & $(1/2)\times4^n-2$  \\
        QSD~\cite{1629135} & \href{https://quantum.cloud.ibm.com/docs/en/api/qiskit/qiskit.transpiler.passes.UnitarySynthesis}{qiskit} & $\Vert A\Vert_2$
            & 20 & 100 & 444 & 1868 & 7660
            & $(23/48)\times4^n - (3/2)\times 2^n + (4/3)$ \\
        Block-ZXZ~\cite{PhysRevApplied.22.034019} & \href{https://github.com/zexianLIPolyU/siable/blob/main/test_siable_CNOT.m}{siable} & $\Vert A\Vert_2$
            & 19 & 95 & 423 & 1783 & 7319
            & $(22/48)\times4^n - (3/2)\times 2^n + (5/3)$ \\
        \textbf{SIABLE (Proposed method)} &\href{https://github.com/zexianLIPolyU/SPDMM-SIABLE}{siable}
            & $\Vert A\Vert_2$
            & \textbf{9} & \textbf{45} & \textbf{205} & \textbf{877} & \textbf{3629}
            & $(11/48)\times 4^n - 2^n + (7/3)$ \\ \hline

        Lower bound on block encoding  (Thm.~\ref{theorem: lower bound of C-NOT for block encoding}) &  & $\Vert A\Vert_2$
            & 6 & 29 & 125 & 508 & 2043
            & $\lceil (1/8)\times4^n - (3/4)\times n \rceil$ \\
        \hline
    \end{tabular}

\end{table*}

\subsection{Low Rank Matrix Encoding}
\label{subsection low rank encoding}

Low-rank matrices are widely used in real-world applications, including recommendation systems~\cite{kerenidis_et_al:LIPIcs.ITCS.2017.49}, large language models~\cite{hu2022lora}, and computer vision~\cite{1580791}. Since low-rank matrices possess fewer degrees of freedom compared to full-rank matrices, we can reduce the C-NOT cost of block encoding for low-rank matrices by designing specialized quantum circuit architectures.
The target upper-left block encoding unitary $U_A$ consists of a quantum multiplexor $A_1 \oplus A_2$, which can be decomposed using Eq.~\eqref{eq_quantum_multiplexor}. $W_{A}$ and $V_{A}^\dagger$ are the left and right-singular vectors of $A$, and $D=\text{diag}(e^{i\theta_0},e^{i\theta_1},\cdots,e^{i\theta_{2^{n-1}-1}})$, $\{\cos(\theta_i)\}$ are the singular values of $A$.
If $A = W_A \Sigma_A V_A^\dagger$ is a rank-$K$ matrix, then $\cos\theta_i = 0$ for $i=K,\cdots,2^{n-1}-1$. Consequently, the synthesis of the first $K$ columns of $W_A$ and the first $K$ rows of $V_A^\dagger$ is adequate for the encoding protocol, which can be implemented through the column-by-column decomposition of isometry~\cite{PhysRevA.93.032318,PhysRevA.71.052330,8blx-nfcr}.

\begin{theorem}[C-NOT counts for encoding low-rank matrix]
\label{theorem:low_rank_siable}
    For a $2^{n-1}\times 2^{n-1}$ matrix $A$ with rank $1\leq K\leq 2^{n-1}-1$, let $K' := 2^{\lceil\log_{2}K\rceil}$ be the smallest power of two not less than $K$, the $(\Vert A\Vert_2, 1)$-block encoding can be implemented with at most $(K'+\frac{11}{12})2^n + \mathcal{O}\left(K'n^2\right)$ C-NOT gates.
\end{theorem}

\begin{proof}
    A rank-$K$ matrix encoding in $\mathbb{C}^{2^{n-1}\times 2^{n-1}}$ is derived from Eq.~\eqref{eq multiplexor decomposition}. It consists of two parts:
    \begin{itemize}
        \item The circuit of isometry~\cite{PhysRevA.93.032318} from $K'$ to $2^{n-1}$ implements the first $K$ rows of $V_A^\dagger$ up to a diagonal matrix $\Delta_{V_A}$, and uses an isometry to implement the first $K$ columns of $W_A\cdot \Delta_{V_A}^\dagger$ with $N_{iso}(K,n-1)=K'2^{n-1} + \mathcal{O}(K'n^2)$ C-NOTs;
        \item The circuit of the uniformly controlled rotations about the $z$ axis takes $N_{{\rm UCRZ}}(n) = 2^{n-1}$ C-NOTs;
    \end{itemize}
    Above all, the C-NOT gates required to synthesize a rank-$K$ matrix are at most
    \begin{equation}
    \begin{aligned}
        N_{\mathrm{BE}}(K,n-1)
        &\leq 2N_{{\rm iso}}(K',n-1)  + N_{{\rm UCRZ}}(n) \\
        &\leq (K' + \frac{11}{12})2^{n} + \mathcal{O}\left(K' n^2\right).  \\
    \end{aligned}
    \label{eq rank-k CNOT}
    \end{equation}
\end{proof}

The SIABLE method is adjustable with the rank. We observe that state preparation is an isometry mapping from $1$ to $2^n$, while unitary synthesis is an isometry mapping from $2^n$ to $2^n$. The isometry synthesis in the low-rank block encoding can be implemented using the state preparation in section~\ref{section state preparation}, the column-by-column isometry synthesis in reference~\cite{PhysRevA.93.032318}, the Block-ZXZ based isometry synthesis in subsection~\ref{subsection: Isometry Synthesis}, or the unitary synthesis in subsection~\ref{subsection: Unitary Synthesis}.
The circuit synthesis of isometry from $K$ to $2^{n-1}$ proceed as follows:
\begin{itemize}
    \item When the rank $K=1$, synthesize the isometry via the state preparation method in section~\ref{section state preparation};
    \item When the rank $K\in[2,2^{n-2}]$,
    \begin{itemize}
        \item if the C-NOT count of the column-by-column isometry ($N_{iso}(k,n-1)$) is less than that of the Block-ZXZ based isometry ($N_{iso\cdot \Delta}(2^{n-2},n-1)$) as $N_{iso}(k,n-1)< N_{iso\cdot \Delta}(2^{n-2},n-1)$, synthesize the isometry via the column-by-column isometry~\cite{PhysRevA.93.032318};
        \item otherwise, synthesize the isometry via the Block-ZXZ based isometry in subsection~\ref{subsection: Isometry Synthesis};
    \end{itemize}
    \item When the rank $K> 2^{n-2}$, synthesize the isometry via the unitary synthesis method in subsection~\ref{subsection: Unitary Synthesis}.
\end{itemize}

The upper bound on the C-NOT in the low-rank single ancilla block encoding is given by
$$
\begin{aligned}
    & N_{\text{BE}}(K,n-1) \le N_{{\rm UCRZ}}(n) + \\
    &\ 2\min\!\Bigl\{\underbrace{N_{{\rm iso}}(K,n-1)}_{\text{column-by-column}}, \underbrace{N_{{\rm iso \cdot \Delta}}(2^{n-2},n-1)}_{\text{Block-ZXZ isometry}}, \underbrace{N_{{\rm u \cdot \Delta}}(n-1)}_{\text{unitary}}\Bigr\}.
\end{aligned}
$$

The single-ancilla block encoding method (SIABLE) proceeds as follows:
\begin{itemize}
    \item \textbf{Phase 0.} Apply singular value decomposition (SVD) to the matrix $A=\sum_{i=0}^{2^n-1} \sigma_iw_i v_i^\dagger =W_A\Sigma_A V_A^\dagger$.
    \item \textbf{Phase 1.} Generate the circuit of an isometry $V_A^\dagger$ mapping from $2^n$ to $\mathrm{rank}(A)$ up to a diagonal matrix $\Delta_{V_A^\dagger}$, with the isometry method we mentioned above. Apply a Hadamard gate.
    \item \textbf{Phase 2.} Generate the circuit of uniformly controlled rotations $\text{diag}(D, D^\dagger)$, where $D = (e^{i\theta_1}, \cdots, e^{i\theta_{2^{n-1}}})$ and $\theta_i = \arccos(\sigma_i)$.
    \item \textbf{Phase 3.} Generate the circuit of an isometry from $\mathrm{rank}(A)$ to $2^n$ as $W_A\cdot\Delta_{V_A^\dagger}$. Apply a Hadamard gate.
\end{itemize}
A schematic overview of the SIABLE protocol is illustrated in Fig.~\ref{fig:block encoding}, and the corresponding C-NOT counts for low- and full-rank matrices are detailed in Table~\ref{tab:C-NOT_count_SIABLE_low_rank}. We also give two concrete examples of $4\times 4$ rank-one and full-rank matrix block encoding protocols in Appendices~\ref{appendix:example_rank_one_block_encoding} and~\ref{appendix:example_full_rank_block_encoding},  respectively.



  \begin{figure}[htbp]
     \centering
     \includegraphics[width=\linewidth]{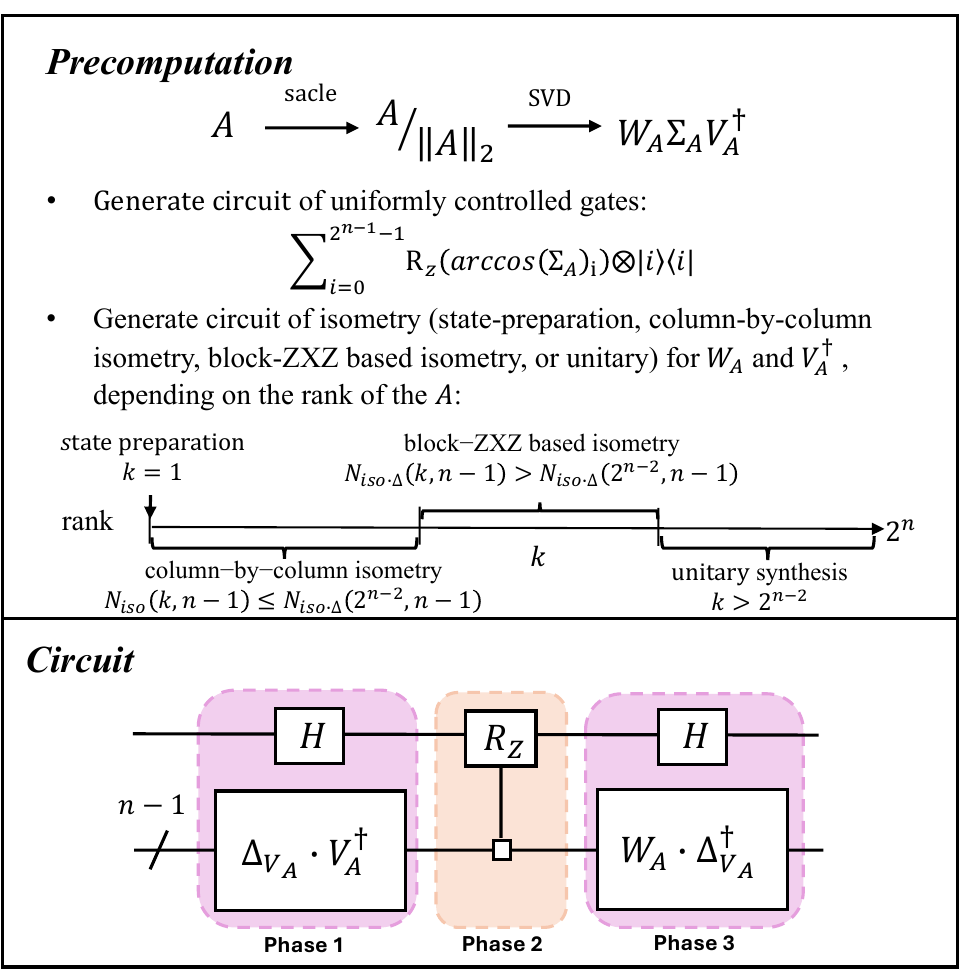}
     \caption{The process of synthesizing the SIABLE. The synthesis of $V_A^\dagger$ up to $\Delta_{V_A}$ will migrate into the synthesis of $W_A\cdot \Delta_{V_A}^\dagger$.}
     \label{fig:block encoding}
 \end{figure}

\begin{table}[htbp]
    \centering
    \caption{Number of C-NOT gates in the SIABLE for general low-rank and full-rank $2^{n-1}\times 2^{n-1}$ matrices.}
    \label{tab:C-NOT_count_SIABLE_low_rank}
    \centering
    \begin{tabular}{|c|ccccccc|}
		\hline
        \makecell[l]{$n\backslash$rank}
        & 1 & 2 & 3 & 4 & 5 & 10 & full-rank \\ \hline
        3 & 6 & & & & & & 9 \\ \hline
        4 & 14 & 28 & 36 & 36 & & & 45 \\ \hline
        5 & 30 & 66 & 130 & 130 & 156 &  & 205 \\ \hline
        6 & 68 & 144 & 276 & 276 & 554 & 660 & 877 \\ \hline
        7 & 148 & 302 & 566 & 566 & 1120 & 2278 & 3629 \\ \hline
        8 & 314 & 620 & 1144 & 1144 & 2230 & 4488 & 14765 \\ \hline
		9 & 654 & 1128 & 2298 & 2298 & 4428 & 8810 & 59565 \\ \hline
	\end{tabular}
\end{table}

\section{Experiments}
\label{section_experiments}



    In this section, an experiment on encoding low-rank matrices is performed to study the performance difference between the single ancilla block encoding protocol (SIABLE) and other block encoding methods.

    In this section, we compare our proposed SIABLE protocol against three block-encoding baselines: (i) FABLE~\cite{9951292} using max-norm normalization, (ii) BITBLE~\cite{11216010} using Frobenius-norm normalization, and (iii) Unitary-synthesis method with optimal normalization factor $\alpha$ — Block-ZXZ~\cite{PhysRevApplied.22.034019}. We summarize the basic information of the above block encoding protocol in Table~\ref{tab:C-NOT_count_SIABLE}.

    \begin{table}[htbp]
    \centering
    \caption{Comparison of the normalization factor $\alpha$, number of ancilla qubits, and the leading constant of C-NOT counts among the proposed single ancilla block encoding protocol (SIABLE) for general $2^{n-1}\times 2^{n-1}$ dense full-rank matrices, other block encoding protocols, and theoretical lower bounds.}
    \label{tab:C-NOT_count_SIABLE}
    \begin{tabular}{|l|c|c|c|}
		\hline
		Methods & $\alpha$ & Ancilla & Leading Constant \\
		\hline
        FABLE~\cite{9951292} & $2^{n-1}\Vert A\Vert_{\infty}$ & $n$ & $1/2$ \\
        BITBLE~\cite{11216010} & $\Vert A\Vert_F$ & $n-1$ & $1/2$ \\

        Block-ZXZ~\cite{PhysRevApplied.22.034019} & $\Vert A\Vert_2$ & $1$ & $11/24$ \\
        \makecell[l]{\textbf{SIABLE}\\ \textbf{\small(Proposed method)}} & $\Vert A\Vert_2$ & $1$ & $\textbf{11}/\textbf{48}$ \\
        \hline
        \makecell[l]{Bound on block\\ encoding  (Thm.~\ref{theorem: lower bound of C-NOT for block encoding})} & $\Vert A\Vert_2$ & $1$ & $1/8$ \\
		\hline
	\end{tabular}
    \end{table}

    In many quantum algorithms~\cite{10.1145/3313276.3316366,PRXQuantum.2.040203}, the normalization factor is proportional to the query complexity.
    We define the effective cost as
    \[
    C_{\mathrm{eff}} = N_{\mathrm{BE}} \cdot \frac{\alpha}{\|A\|_2},
    \]
    where $N_{\mathrm{BE}}$ denotes the C-NOT count of the block encoding and $\alpha$ is the normalization factor.

    Figure~\ref{fig:be-eff-rank} fixes $n=7$ and varies the rank $K$ from 1 to full rank, with all methods encoding the same matrix at each value of $K$. For each rank $K$, we sample a random rank-$K$ matrix $A = U \operatorname{diag}(\boldsymbol{\sigma}) V^{\top}$, in which $U, V \in \mathbb{R}^{2^{n-1}\times K}$ are matrices with orthonormal columns, obtained via QR factorization of Gaussian random matrices, and the singular values $\sigma_i$ are sampled uniformly from $[0.2, 1]$ prior to normalization. Each matrix is subsequently normalized to unit spectral norm, $A \leftarrow A/\|A\|_2$, and every reported data point is averaged over $12$ independent realizations. Since the Block-ZXZ~\cite{PhysRevApplied.22.034019} method synthesizes a full unitary, its effective cost is independent of $K$ (the two curves nearly coincide). In contrast, the rank-adaptive construction of SIABLE scales as $(2^{\lceil\log_2 K\rceil} + \tfrac{11}{12})2^n$ for $K \in [2, 2^{n-2}]$. For genuinely low-rank inputs, SIABLE is more than an order of magnitude cheaper---for example, at $K=1$ it uses only $148$ C-NOTs compared to $7319$ for Block-ZXZ. Since the normalization factor $\alpha$ of BITBLE and FABLE method increases with rank, their effective cost also increases with rank. Across the entire range $1 \le K \le 2^{n-1}-1$, the SIABLE method remains the most efficient protocol.

\begin{figure}[htbp]
    \includegraphics[width=0.99\linewidth]{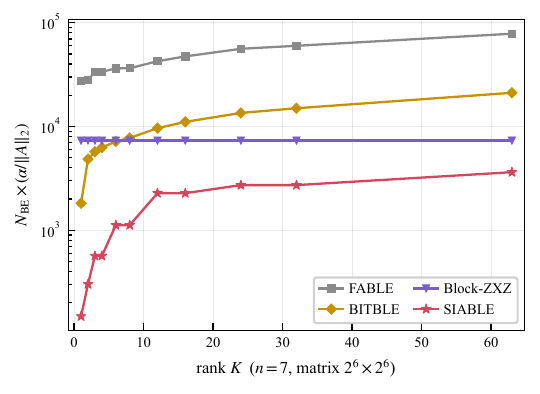}
    \caption{Comparison of different block-encoding protocols on $2^{6} \times 2^{6}$ matrices with varying ranks, the $y$-coordinate represents the numbers of C-NOT gates ($N_{\text{BE}}$) multiplied by the normalization factor over the spectral norm ($\alpha/\Vert A\Vert_2$). }
    \label{fig:be-eff-rank}
\end{figure}

\section{Conclusion}
\label{section_conclusion}
In this article, we introduce state preparation method via diagonal matrix migration (SPDMM) and single ancilla block encoding protocol (SIABLE).  Techniques like diagonal matrix migration, isometry synthesis, recursion, and data processing have been applied to optimize quantum circuits. Our algorithms, as summarized in Tables~\ref{tab:C-NOT_count state preparation}-\ref{tab: SIABLE}, achieve the lowest known C-NOT gate counts currently known. Besides, we propose the novel block encoding protocol for low-rank matrices. 
Specifically, the SIABLE protocol reduces the C-NOT count by over 50\% for block encoding synthesis compared to the previous unitary synthesis method.

Furthermore, we have established a lower bound on the number of C-NOT gates required for both state preparation and block encoding, based on circuit topology. Despite this progress, a gap remains between the performance of current algorithms—which achieve an 11/12 ratio for state preparation (or 11/48 for block encoding)—and the theoretical lower bound of 1/2 for state preparation (or 1/8 for block encoding). Future work could therefore explore methods to further bridge the gap between worst-case algorithm performance and this fundamental limit.


\section*{Data availability}

The sites \url{https://github.com/zexianLIPolyU/SPDMM-SIABLE} contain the software in the current study.

\section*{Acknowledgements}
The authors used Grok and Claude to enhance the grammatical accuracy and overall language quality of the manuscript. AI tools also assisted in generating the Python code presented in Appendices E and F. After using these tools, the authors thoroughly reviewed and edited the content as needed and accept full responsibility for the accuracy, integrity, and scientific validity of the final paper. The tools played no role in the generation of ideas, data analysis, or substantive content.

\appendices
\section{Decoupling of A $\mathrm{U}(4)$ Operation into single-qubit Gates and Two C-NOT Gates}
\label{appendix: Decoupling two-qubit operations}
This subsection follows the reference~\cite{PhysRevA.69.062321}.
A unitary matrix in the Lie group $\mathrm{U}(4)$ can be represented as the product of a global phase factor and an element of the special unitary group $\mathrm{SU}(4)$. For $u\in \mathrm{U}(4)$, denote $\gamma(u) = u\sigma_y^{\otimes 2}u^T\sigma_y^{\otimes 2}$, and denote $\chi[u]$ as the characteristic polynomial of $u$. For any unitary matrix $U\in\mathrm{SU}(4)$, Fig.~\ref{quantum_circuit_SU4} shows that there exist $\psi,\theta$ and $\phi$ such that $\Delta U$ and $CX \left(R_X(\theta)\otimes R_Z(\phi)\right)CX$ are in the same double coset of $\mathrm{SU}(4)$ modulo $\mathrm{SU}(2)^{\otimes 2}$, which also implies that the characteristic polynomials of the two operations are equivalent~\cite{PhysRevA.69.062321} as
\begin{equation}
    \chi\left[\gamma(\Delta U)\right] = \chi\left[\gamma(CX \left(R_X(\theta)\otimes R_Z(\phi)\right)CX)\right],
    \label{eq_equivalence}
\end{equation}
where $\Delta = {\rm diag}(e^{-i\frac{\psi}{2}},e^{i\frac{\psi}{2}},e^{i\frac{\psi}{2}},e^{-i\frac{\psi}{2}})$, $CX = \ket{0}\bra{0}\otimes I_2 + \ket{1}\bra{1}\otimes X$.

\begin{figure}[htbp]
$$
\begin{myqcircuit*}{1}{0.7}
		& \qw & \multigate{1}{U} & \multigate{1}{\Delta} & \qw \\
		& \qw & \ghost{U} & \ghost{\Delta}  & \qw \\
\end{myqcircuit*}
=
\begin{myqcircuit*}{0.5}{0.7}
		& \qw & \gate{a} & \ctrl{1} & \gate{R_X(\theta)} & \ctrl{1} & \gate{c} & \qw \\
		& \qw & \gate{b} & \targ & \gate{R_Z(\phi)} & \targ & \gate{d} & \qw \\
\end{myqcircuit*}
$$
\caption{Quantum circuit of $\mathrm{SU}(4)$ decomposition with $2$ C-NOT gates, where $a,b,c,d\in\mathrm{SU}(2)$ are four one-qubit gates, and $\Delta = {\rm diag}(e^{-i\frac{\psi}{2}},e^{i\frac{\psi}{2}},e^{i\frac{\psi}{2}},e^{-i\frac{\psi}{2}})$ is a diagonal matrix in $\mathrm{SU}(4)$.}
\label{quantum_circuit_SU4}
\end{figure}

 The decoupling procedures within the circuit depicted in Fig.~\ref{quantum_circuit_SU4} can be systematically determined through a two-step process:

Step 1. Determine the parameters $\psi,\theta$ and $\phi$.

For any $M\in\mathrm{SU}(4)$, the coefficients of $\chi\left[M\right]$ are all real only when $\text{tr}\left[M\right]$ is real. Denote $\text{diag}(\gamma(U)) = [t_1,t_2,t_3,t_4]$. On one hand, $\text{tr}\left[ \gamma(\Delta U)\right] = \text{tr}\left[\Delta ^2 \gamma(U)\right] = (t_1+t_4)e^{-i\psi} + (t_2+t_3)e^{i\psi}$ is real implies $\psi = \arctan\left(\frac{\mathrm{Im}(t_1+t_2+t_3+t_4)}{\mathrm{Re}(t_1-t_2-t_3+t_4)}\right)$. On the other hand, with $\psi$ determined above, the roots of the characteristic polynomial can be paired in conjugate
pairs as $\chi\left[ \gamma(\Delta U)\right] = \chi\left[ \gamma(CX\left(R_X(\theta)\otimes R_Z(\phi)\right)CX)\right] = (X - e^{is})(X - e^{-is})(X - e^{it})(X - e^{-it})$, which implies that $\theta = (r+s)/2$ and $\phi = (r-s)/2$.

Step 2. Determine the single-qubit operations $a,b,c,d \in \mathrm{SU}(2)$ such that $\Delta U = (c\otimes d)CX\left(R_X(\theta)\otimes R_Z(\phi)\right)CX(a\otimes b)$.

Denote
$$
E = \frac{1}{\sqrt{2}}\begin{pmatrix}
    1 & i & 0 & 0 \\
    0 & 0 & i & 1 \\
    0 & 0 & i & -1 \\
    1 & -i & 0 & 0 \\
\end{pmatrix}.
$$
For two operations $u,v\in \mathrm{SU}(4)$~\cite{artin2011algebra,PhysRevLett.91.027903},
\begin{equation}
\begin{aligned}
    &\chi[\gamma(u)] = \chi[\gamma(v)] \\
    \Leftrightarrow & \chi[E^\dagger \gamma(u) E] = \chi[E^\dagger \gamma(v) E] \\
    \Leftrightarrow &\chi[E^\dagger u \sigma_y^{\otimes 2} u^T \sigma_y^{\otimes 2} E] = \chi[E^\dagger v \sigma_y^{\otimes 2} v^T \sigma_y^{\otimes 2} E] \\
    \Leftrightarrow & \chi[E^\dagger u E E^Tu^T{E^T}^\dagger E^\dagger E] = \chi[E^\dagger v E E^Tv^T{E^T}^\dagger E^\dagger E] \\
    \Leftrightarrow &\chi[(E^\dagger u E)(E^\dagger u E)^T] = \chi[(E^\dagger v E)(E^\dagger v E)^T], \\
\end{aligned}
\label{eq_equivalence2}
\end{equation}
where the first equivalence derives from the point that $E\gamma E^\dagger$ does not change the equivalence of $\chi[\gamma(u)] = \chi[\gamma(v)]$~\cite{PhysRevA.69.062321}, the second equivalence derives from the definition of $\gamma(\cdot)$, the third equivalence derives from $-\sigma_y^{\otimes 2}=EE^T=(EE^T)^\dagger$. Given $U\in \rm{SU}(4)$, for $\omega_1 = E^\dagger (\Delta U ) E$ and $\omega_2 = E^\dagger (CX\left(R_X(\theta)\otimes R_Z(\phi)\right)CX) E$, we have
$$
\begin{aligned}
    &\chi[\gamma(\Delta U)] = \chi[\gamma(CX\left(R_X(\theta)\otimes R_Z(\phi)\right)CX)] \\
    \Leftrightarrow & \chi[\omega_1\omega_1^T]=\chi[\omega_2\omega_2^T] \\
    \Leftrightarrow & \exists s_1,s_2\in{\rm SO}(4), s_1\omega_1\omega_1^Ts_1^T = s_2\omega_2\omega_2^Ts_2^T \\
    \Leftrightarrow & (\omega_2^\dagger s_2^Ts_1^T \omega_1)(\omega_2^\dagger s_2^Ts_1^T \omega_1)^T = I \\
\end{aligned}
$$
where the first equivalence derives from Eq.~\eqref{eq_equivalence2}, and the second equivalence derives from the spectral theorem~\cite{PhysRevA.69.062321}. Denote $s_3 = \omega_2^\dagger s_2^Ts_1^T \omega_1$, the above last equivalence implies that $s_3\in \mathrm{SO}(4)$. Besides, let $\omega_1 = s_1s_2\omega_2 s_3$,
$$
\begin{aligned}
    \Delta U &= E \omega_1 E^\dagger\\
    &= (E s_1s_2 E^\dagger) (CX \left(R_X(\theta)\otimes R_Z(\phi)\right) CX) (E s_3 E^\dagger).
\end{aligned}
$$
Finally, based on $E\cdot \mathrm{SO}(4)\cdot E^{\dagger}= \rm{SU(2)}^{\otimes 2}$ and $s_1,s_2,s_3\in{\rm SO}(4)$, the objective single-qubit unitaries can be decoupled from $c\otimes d = E s_1s_2 E^\dagger$ and $a\otimes b = E s_3 E^\dagger$.

Above all, the number of C-NOT gates for synthesizing a $2$-qubit unitary $U \in \mathrm{U}(4)$ up to a diagonal matrix is $N_{u\cdot \Delta}(2) = 2$. It also implies that the number of C-NOT gates for synthesizing a $2$-qubit isometry from $2$ to $4$ is $N_{\text{iso} \cdot \Delta}(2,2) = 2$.

\section{Example of state preparation method via diagonal matrix migration (SPDMM) for a $3$-qubit state}\label{appendix:example}

Consider a 3-qubit state
$$
\ket{\psi} = [\psi_0, \psi_1, \psi_2, \psi_3, \psi_4, \psi_5, \psi_6, \psi_7]^{\mathrm{T}}.
$$

The state preparation proceeds as follows:

Steps 1 and 2 (Reshape and SVD). Reshape the state vector into a matrix and perform its singular value decomposition (SVD):
$$
\psi = \begin{bmatrix}
\psi_0 & \psi_1 \\
\psi_2 & \psi_3 \\
\psi_4 & \psi_5 \\
\psi_6 & \psi_7
\end{bmatrix} = U_{4\times 4} \, \Sigma_{4\times 2} \, (V_{2\times 2})^\dagger.
$$

Step 3 (Circuit synthesis for unitary components). Construct the quantum circuit for \( U \cdot \Delta_U \) using the method described in Appendix~\ref{appendix: Decoupling two-qubit operations}. Similarly, synthesize the circuit for \( \operatorname{conj}(V) \) and append the required C-NOT gate:
$$
\begin{myqcircuit*}{0.2}{0.7}
& & & & \mbox{$U \cdot \Delta_U \cdot (\det(\operatorname{conj}(V))^{-1/4}$} & & & & \\
& \qw & \gate{a} & \ctrl{1} & \gate{R_X(\theta)} & \ctrl{1} & \gate{c} & \qw \\
& \ctrl{1} & \gate{b} & \targ & \gate{R_Z(\phi)} & \targ & \gate{d} & \qw \\
& \targ & \gate{v} & \qw & \qw & \qw & \qw & \qw  \gategroup{2}{3}{3}{7}{0.95em}{--}
\gategroup{4}{3}{4}{3}{0.5em}{--}
\end{myqcircuit*}
$$
where \( v = \operatorname{conj}(V) \cdot (\sqrt{\det(\operatorname{conj}(V))})^{-1} \in \operatorname{SU}(2) \).

Step 4 (Recursive preparation of the residual state). Recursively prepare the state $\ket{\psi'} = \operatorname{diag}(\Delta_U^\dagger \Sigma) \cdot (\det(\operatorname{conj}(V))^{1/4} \sqrt{\det(\operatorname{conj}(V))}  = [\psi_0', \psi_1']^{\mathrm{T}}$ via the unitary
$$
U_{\psi'} = \begin{bmatrix}
\psi_0' & -\operatorname{conj}(\psi_1') \\
\psi_1' & \operatorname{conj}(\psi_0')
\end{bmatrix}.
$$

Combining the above steps yields the complete quantum circuit for generating \( \ket{\psi} \):
$$
\begin{myqcircuit*}{0.2}{0.2}
& \qw & \qw & \gate{a} & \ctrl{1} & \gate{R_X(\theta)} & \ctrl{1} & \gate{c} & \qw \\
& \gate{U_{\psi'}} & \ctrl{1} & \gate{b} & \targ & \gate{R_Z(\phi)} & \targ & \gate{d} & \qw \\
& \qw & \targ & \gate{v} & \qw & \qw & \qw & \qw & \qw  \\
\end{myqcircuit*}
$$

\section{Example of single ancilla block encoding (SIABLE) for a rank-one $4\times 4$ matrix}
\label{appendix:example_rank_one_block_encoding}

Consider a rank-one matrix $A\in\mathbb{C}^{4\times 4}$ written in its single-term singular value decomposition,
\begin{equation}
    A = \sigma_{1} w_{1}v_{1}^{\dagger} \in\mathbb{C}^{4\times 4},
\label{eq:appC-A}
\end{equation}
with $\sigma_1\leq 1$. The SIABLE construction proceeds as follows:

Phase 1. Generate the inverse state preparation circuit $V_A^\dagger$ with 1 C-NOT, where $V_A\ket{0}^{\otimes 2} = \ket{v_1}$ and append a Hadamard gate:
$$
\mbox{$V_A^\dagger$}\quad
\begin{myqcircuit*}{0.2}{0.35}
& \gate{H}   & \qw      & \qw        & \qw \\
& \gate{U} & \ctrl{1} & \gate{U} & \qw \\
& \gate{U} & \targ    & \qw        & \qw  \gategroup{2}{2}{3}{4}{0.5em}{--}
\end{myqcircuit*}
$$
where the (inverse) state preparation circuit $V_A^\dagger$ is generated using the SPDMM method.

Phase 2. Generate the circuit of uniformly controlled rotations $\text{UCR}_Z = \operatorname{diag}\!\bigl(D, D^\dagger\bigr)$ with 4 C-NOTs, where $D = \operatorname{diag}\!\bigl(e^{i\theta_{1}},\,i,\,i,\,i\bigr)$. The first entry $e^{i\theta_{1}}$ carries the singular value via $\tfrac{1}{2}\!\left(e^{i\theta_{1}}+e^{-i\theta_{1}}\right)=\cos\theta_{1}=\sigma_{1}$,
the remaining three entries are pinned to $i=e^{i\arccos 0}$ so that $\tfrac{1}{2}(i+(-i))=0$, leading to $\theta_i = \pi/2$ for $i=2,3,4$.
$$
\begin{myqcircuit*}{0.2}{0.35}
& & & & & & &  \mbox{$\text{UCR}_Z$} & & & \\
& \gate{H} & \qw & \qw & \gate{R_z} & \targ & \gate{R_z} & \targ & \gate{R_z} & \targ & \gate{R_z} & \targ & \qw \\
& \gate{U} & \ctrl{1} & \gate{U} & \qw & \ctrl{-1} & \qw & \qw & \qw & \ctrl{-1} & \qw & \qw & \qw \\
& \gate{U} & \targ & \qw & \qw & \qw & \qw & \ctrl{-2} & \qw & \qw & \qw & \ctrl{-2} & \qw  \gategroup{2}{5}{4}{13}{0.5em}{--}
\end{myqcircuit*}
$$
The angles in the uniformly controlled rotations $\text{UCR}_Z$ are decoupled via Walsh–Hadamard transformation reordered to Gray-code indexing~\cite{PhysRevLett.93.130502}.

Phase 3. Generate the state preparation circuit $W_A$ using the SPDMM method with 1 C-NOT, where $W_A\ket{0}^{\otimes 2} = \ket{w_1}$ and append a Hadamard gate:
$$
\begin{myqcircuit*}{0.2}{0.2}
& \gate{H} & \qw & \gate{R_z} & \targ & \gate{R_z} & \targ & \gate{R_z} & \targ & \gate{R_z} & \targ & \gate{H} & \qw & \qw & \qw \\
& \gate{U} & \ctrl{1} & \gate{U} & \ctrl{-1} & \qw & \qw & \qw & \ctrl{-1} & \qw & \qw & \gate{U} & \ctrl{1} & \gate{U} & \qw \\
& \gate{U} & \targ & \qw & \qw & \qw & \ctrl{-2} & \qw & \qw & \qw & \ctrl{-2} & \qw & \targ & \gate{U} & \qw \gategroup{2}{12}{3}{14}{0.5em}{--}
\end{myqcircuit*}
\quad\mbox{$W_A$}
$$

\section{Example of Single Ancilla Block Encoding (SIABLE) for a Full-Rank $4 \times 4$ Matrix}
\label{appendix:example_full_rank_block_encoding}

Consider a full-rank matrix $A \in \mathbb{C}^{4\times 4}$ with spectral norm $\|A\|_2 = 1$. Apply singular value decomposition to obtain
\begin{equation*}
A = W_A \Sigma_A V_A^\dagger,
\end{equation*}
where $W_A, V_A \in \mathrm{U}(4)$ are 2-qubit unitaries and
\begin{equation*}
\Sigma_A = \mathrm{diag}(\sigma_1, \sigma_2, \sigma_3, \sigma_4)
\end{equation*}
is the diagonal matrix of singular values with $\sigma_i \leq 1$. The SIABLE construction for full-rank matrix proceeds as follows.

Phase 1. Generate the circuit of $V_A^\dagger$ up to a diagonal matrix $\Delta_{V_A}$. Since $V_A^\dagger$ is a 2-qubit unitary in $\mathrm{U}(4)$, it is synthesized up to a diagonal using $2$ C-NOTs and four single-qubit gates $a_1, b_1, c_1, d_1 \in \mathrm{SU}(2)$ as stated in Fig.~\ref{quantum_circuit_SU4} of Appendix~\ref{appendix: Decoupling two-qubit operations}. Append a Hadamard gate on the ancilla:
$$
\mbox{$\Delta_{V_A}\cdot V_A^\dagger$}\quad
\begin{myqcircuit*}{0.2}{0.2}
        & \qw & \gate{H} & \qw & \qw & \qw & \qw & \qw \\
		& \qw & \gate{a} & \ctrl{1} & \gate{R_X} & \ctrl{1} & \gate{c} & \qw \\
		& \qw & \gate{b} & \targ & \gate{R_Z} & \targ & \gate{d} & \qw \gategroup{2}{2}{3}{7}{0.5em}{--}
\end{myqcircuit*}
$$
Phase 2. Generate the uniformly controlled rotation
$\mathrm{UCR}_Z = \mathrm{diag}(D, D^\dagger)$ with 4 C-NOTs,
where $D = \mathrm{diag}(e^{i\theta_1}, e^{i\theta_2}, e^{i\theta_3}, e^{i\theta_4})$, $\theta_i = \arccos\sigma_i$, for $i=1,2,3,4$. The rotation angles of rotation angles are decoupled by a Walsh–Hadamard transformation~\cite{PhysRevLett.93.130502}.
$$
\begin{myqcircuit*}{0.2}{0.2}
& & & & & & & & &  \mbox{$\text{UCR}_Z$} & & & \\
& \gate{H} & \qw & \qw & \qw & \qw & \gate{R_z} & \targ & \gate{R_z} & \targ & \gate{R_z} & \targ & \gate{R_z} & \targ & \qw \\
& \gate{a} & \ctrl{1} & \gate{R_X} & \ctrl{1} & \gate{c} & \qw & \ctrl{-1} & \qw & \qw & \qw & \ctrl{-1} & \qw & \qw & \qw \\
& \gate{b} & \targ & \gate{R_Z} & \targ & \gate{d} & \qw & \qw & \qw & \ctrl{-2} & \qw & \qw & \qw & \ctrl{-2} & \qw  \gategroup{2}{7}{4}{15}{0.5em}{--}
\end{myqcircuit*}
$$

Phase 3. Generate the circuit of $W_A \cdot \Delta_{V_A}^\dagger$ with $3$ C-NOTs using the unitary method for $\mathrm{SU}(4)$ stated in~\cite{PhysRevA.69.062321}, and append a Hadamard gate on the ancilla.
$$
    \begin{myqcircuit*}{0.2}{0.0}
    & \gate{H} & \qw & \qw & \qw & \gate{R_z} & \targ & \gate{R_z} & \targ & \gate{R_z} & \targ & \gate{R_z} & \targ & \gate{H} & \qw & \qw & \qw & \qw & \qw & \qw & \qw \\
    & \gate{a} & \ctrl{1} & \gate{R_X} & \ctrl{1} & \gate{c} & \ctrl{-1} & \qw & \qw & \qw & \ctrl{-1} & \qw & \qw & \qw & \ctrl{1} & \gate{U} & \ctrl{1} & \gate{U} & \ctrl{1} & \gate{U} & \qw \\
    & \gate{b} & \targ & \gate{R_Z} & \targ & \gate{d} & \qw & \qw & \ctrl{-2} & \qw & \qw & \qw & \ctrl{-2} & \gate{U} & \targ & \gate{U} & \targ & \gate{U} & \targ & \gate{U} & \qw  \\
    & & & & & & & & & & & & & & & & \mbox{$W_A\cdot \Delta_{V_A}^\dagger$} \gategroup{2}{14}{3}{20}{0.5em}{--}
    \end{myqcircuit*}
$$



\section{Compile-time Cost of SPDMM}

To assess the practical compile-time cost of SPDMM, we benchmark its end-to-end circuit construction against five state-preparation baselines: PB~\cite{PhysRevA.83.032302}, low-rank~\cite{10.1109/TCAD.2023.3297972}, Isometry~\cite{PhysRevA.93.032318}, PB+BlockZXZ~\cite{PhysRevApplied.22.034019}, and low-rank+BlockZXZ using Qiskit's compiled kernel. Compile time is measured as the sum of build time (gate emission to a Qiskit's \texttt{QuantumCircuit}) and transpile time (lowering to the $\{\text{SU}(2),\,\text{C-NOT}\}$ basis). All experiments were performed on a Windows PC with an Intel Core i7-12700 CPU and 32 GB RAM. Fig.~\ref{fig:benchmark} reports the compile time of state preparation comparison for $n = 5, \ldots, 15$, and the case of $n=15$ is detailed in Table~\ref{tab:compile_n15}.

In the Qiskit package, Rust-based acceleration has been applied to optimize specific circuit synthesis routines. To ensure a fair comparison, we evaluated the proposed spdmm algorithm using two different implementations.

(a) Pure-Python implementation (spdmm): A pure-Python implementation of the spdmm algorithm, used to verify correctness against Table~\ref{tab:C-NOT_count state preparation}.

(b) Hybrid implementation (spdmm (Rust)): This version delegates the inner unitary synthesis (Step 3 of the SPDMM algorithm) to Qiskit’s compiled Rust kernel while omitting diagonal matrix migration (Step 4). Based on the Rust kernel, the program compiles a $15$-qubit state in 980 ms—approximately twice as fast as the pure-Python version and 15\% faster than the closest C-NOT competitor, PB+BlockZXZ, in the ablation experiment.

In practice, the computational cost of implementing Step~4 is negligible ($\mathcal{O}(2^{n/2})$), in contrast to the substantially higher $\mathcal{O}(2^{3n/2})$ cost of Step~3.
The Rust-based version presented in this work directly leverages Qiskit’s existing interfaces and the C-NOT counts reported for the pure-Python reference in Table~\ref{tab:C-NOT_count state preparation} are language-independent.

\begin{figure}[htbp]
    \centering
    \includegraphics[width=0.9\linewidth]{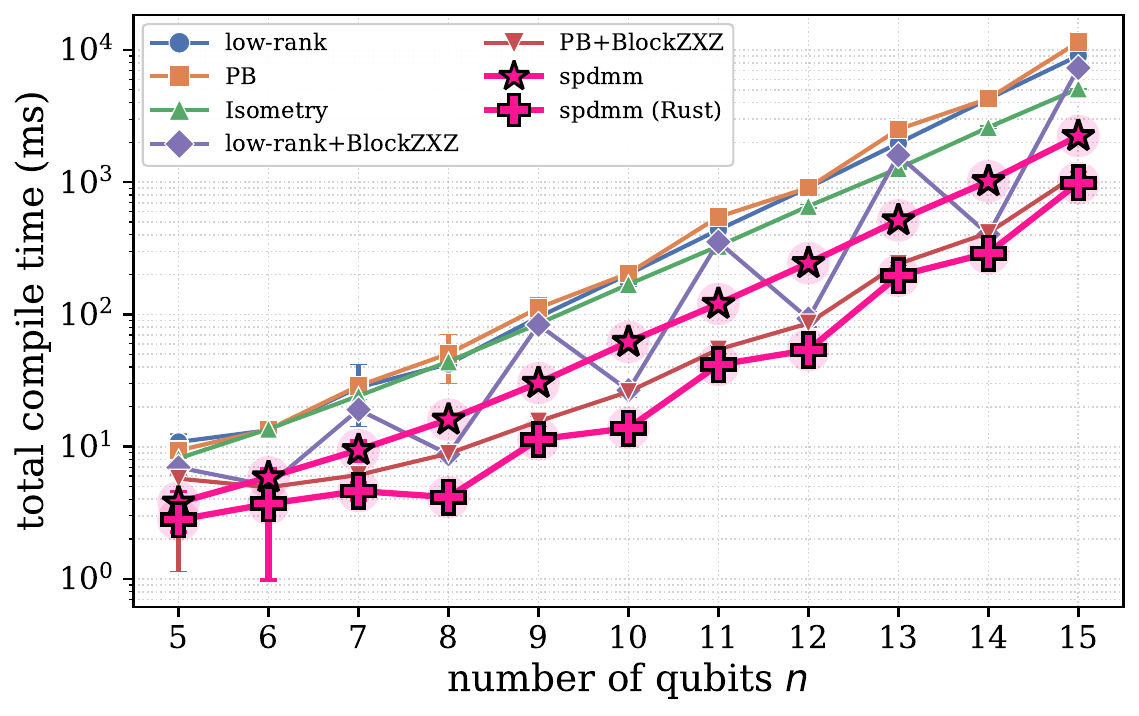}
    \caption{Compile time of state preparation with CI = 95\%.}
    \label{fig:benchmark}
\end{figure}

\begin{table}[htbp]
  \centering
  \caption{Compile time and C-NOT count at $n=15$ for all methods in Fig.\mbox{~\ref{fig:benchmark}}. Compile time is reported as mean $\pm$ $95\%$ CI over $5$ repetitions on Qiskit.}
  \label{tab:compile_n15}
  \renewcommand{\arraystretch}{1.15}
  \begin{tabular}{@{}l r@{${}\pm{}$}l r@{}}
    \toprule
    Method                   & \multicolumn{2}{c}{Compile time (ms)} & C-NOTs           \\
    \midrule
    PB                       & $11342$ & $176$ & $38814$          \\
    low-rank                 & $8922$  & $105$ & $31000$          \\
    low-rank\,$+$\,BlockZXZ  & $7273$  & $241$ & $30659$          \\
    Isometry                 & $4947$  & $44$  & $32752$          \\
    PB\,$+$\,BlockZXZ        & $1147$  & $20$  & $37108$          \\
    \textbf{spdmm}                 & $2207$  & $25$  & $\mathbf{29627}$ \\
    \textbf{spdmm (Rust)}      & $\mathbf{980}$ & $\mathbf{42}$ & $29634$          \\
    \bottomrule
\end{tabular}
\end{table}

\section{Experiment on State Preparation Via Diagonal Matrix Migration}

\begin{figure*}[htbp]
    \centering
    \includegraphics[width=0.9\linewidth]{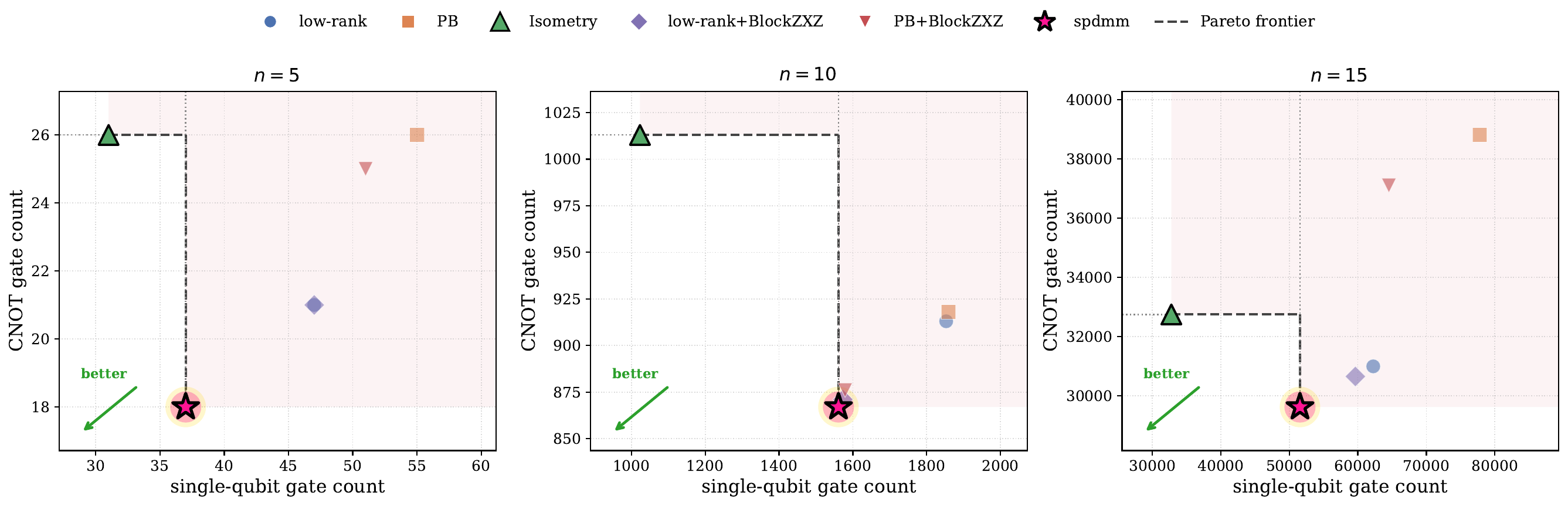}
    \caption{The Pareto frontier of single-qubit counts and C-NOT counts of different state preparation methods for qubit number $n=5,10,15$. The shaded red region indicates the area of the objective space strictly dominated by the Pareto frontier; any point inside this region has both higher single-qubit count and higher C-NOT count than at least one Pareto-optimal method. }
    \label{fig: single-qubit and CNOT's Pareto frontier}
\end{figure*}

    We perform an ablation experiment studying the trade-off between the single-qubit and C-NOT gate counts for state preparation, shown in Fig.~\ref{fig: single-qubit and CNOT's Pareto frontier}.

    We define the relative weighted error of a state-preparation circuit as $W(k,u;r) = k\cdot r + u$, where $k$ and $u$ are the numbers of C-NOT and single-qubit gates after transpilation to the $\{\text{SU}(2),\text{C-NOT}\}$ basis, and $r = \varepsilon_{2Q}/\varepsilon_{1Q}$ is the ratio of two-qubit to single-qubit physical gate error rates. Under independent depolarising noise, $\varepsilon_{1Q}\cdot W(k,u;r)$ is the leading-order expected number of gate-level errors, so the method with the lowest $W(k,u;r)$ at a given $r$ minimises infidelity on a platform with that ratio. On current hardware, $r$ is consistently an order of magnitude above unity---approximately $31$ on IBM Eagle r3~\cite{AbuGhanem2025} ($\varepsilon_{1Q}=2.4\times 10^{-4}$, $\varepsilon_{2Q}=7.6\times 10^{-3}$) and $74$ on Quantinuum H2~\cite{PhysRevX.13.041052} ($\varepsilon_{1Q}=2.5\times 10^{-5}$, $\varepsilon_{2Q}=1.84\times 10^{-3}$)---and reaches $\sim 560$ for the best-demonstrated gate fidelities to date ($\varepsilon_{1Q}=1.5\times10^{-7}$~\cite{42w2-6ccy}, $\varepsilon_{2Q}=8.4\times10^{-5}$~\cite{hughes2025trappediontwoqubitgates9999}). Two-qubit gates therefore dominate the physical error budget, and reducing the C-NOT count is the primary route to lower circuit infidelity.  Table~\ref{tab:weighted} reports the relative weighted error evaluated at $n=15$ for the methods in Fig.~\ref{fig: single-qubit and CNOT's Pareto frontier} using the ratio of two-qubit to single-qubit physical gate error rates in the mentioned hardware, and the SPDMM method has the better relative weighted error compared to the other methods.

\begin{table}[htbp]
    \caption{Relative weighted error $W(k,u;r) = k\cdot r + u$ at $n = 15$ for representative hardware error ratios. Denote C-NOT counts as $k$ and single-qubit counts as $u$.}
    \label{tab:weighted}
    \centering
    \footnotesize
    \setlength{\tabcolsep}{4pt}
    \begin{tabular}{lcccc}
    \toprule
    Method & $k$ & $u$ & $W(r=31)$ & $W(r=74)$ \\
    \midrule
    PB~\cite{PhysRevA.83.032302}                & 38{,}814 & 77{,}801 & 1{,}280{,}035 & 2{,}950{,}037 \\
    PB+Block-ZXZ~\cite{PhysRevApplied.22.034019}      & 37{,}108 & 64{,}529 & 1{,}214{,}877 & 2{,}810{,}521\\
    Isometry~\cite{PhysRevA.93.032318}     & 32{,}752 & 32{,}767 & 1{,}048{,}079 & 2{,}456{,}415 \\
    low-rank~\cite{10.1109/TCAD.2023.3297972} & 30{,}999 & 62{,}246 &
    1{,}023{,}215 & 2{,}356{,}172 \\
    low-rank+Block-ZXZ~\cite{PhysRevApplied.22.034019} & 30{,}658 & 59{,}642 & 1{,}010{,}040 & 2{,}328{,}334 \\
    SPDMM             & 29{,}627 & 51{,}548 & \textbf{970{,}985} & \textbf{2{,}243{,}946} \\
    \bottomrule
    \end{tabular}
\end{table}

\bibliographystyle{IEEEtran}
\bibliography{references}


\end{document}